\documentclass[aps,onecolumn,10pt,prc,floatfix,showpacs,preprintnumbers,amsmath,amssymb,nofootinbib,superscriptaddress,notitlepage]{revtex4-1}

\usepackage[utf8]{inputenc} %
\usepackage[T1]{fontenc}    %
\usepackage{url}            %
\usepackage{hyperref}       %
\usepackage{booktabs}       %
\usepackage{amsfonts}       %
\usepackage{nicefrac}       %
\usepackage{microtype}      %
\usepackage{xcolor}         %
\usepackage{lipsum}
\usepackage{comment}
\usepackage{subcaption}
\usepackage{enumitem}
\usepackage{makecell}

\usepackage{graphicx} %

\usepackage{parskip}

\usepackage{amsmath,amssymb,amsthm}
\usepackage{xspace}
\usepackage{bm}
\usepackage{mathtools}

\usepackage{cleveref}

\newcommand\figsubref[2]{\hyperref[#1]{\ref*{#1}#2}}

\DeclareMathOperator{\diag}{diag}
\DeclareMathOperator{\Tr}{Tr}
\newcommand{\defeq}{\vcentcolon=}
\newcommand{\Minfo}{$\mathcal{M}$-information\xspace}
\newcommand{\M}{\ensuremath{\mathcal{M}}\xspace} 
\newcommand{\Winfo}{$\mathcal{W}$-information\xspace}
\newcommand{\W}{\ensuremath{\mathcal{W}}\xspace}

\usepackage{xparse}

\NewDocumentCommand{\posdef}{o}{%
  \IfNoValueTF{#1}
    {\ensuremath{S_n^{++}}}
    {\ensuremath{S_{#1}^{++}}}%
}
\newcommand{\RealX}[1]{\ensuremath{\mathbb{R}^{#1\times #1}}}

\newcommand{\bV}{\ensuremath{\bm V\xspace}}
\newcommand{\bX}{\ensuremath{\bm X\xspace}}
\newcommand{\bY}{\ensuremath{\bm Y\xspace}}
\newcommand{\bZ}{\ensuremath{\bm Z\xspace}}

\newcommand{\phiid}{\ensuremath{\Phi\mathrm{ID}}\xspace}

\newcommand{\SX}{\ensuremath{\Sigma_X}}
\newcommand{\SY}{\ensuremath{\Sigma_Y}}

\newcommand\rtr{\ensuremath{{R}\!\rightarrow\!{R}}\xspace}
\newcommand\rtx{\ensuremath{{R}\!\rightarrow\!{U}_1}\xspace}
\newcommand\rty{\ensuremath{{R}\!\rightarrow\!{U}_2}\xspace}
\newcommand\rts{\ensuremath{{R}\!\rightarrow\!{S}}\xspace}
\newcommand\xtr{\ensuremath{{U}_1\!\rightarrow\!{R}}\xspace}
\newcommand\xtx{\ensuremath{{U}_1\!\rightarrow\!{U}_1}\xspace}
\newcommand\xty{\ensuremath{{U}_1\!\rightarrow\!{U}_2}\xspace}
\newcommand\xts{\ensuremath{{U}_1\!\rightarrow\!{S}}\xspace}
\newcommand\ytr{\ensuremath{{U}_2\!\rightarrow\!{R}}\xspace}
\newcommand\ytx{\ensuremath{{U}_2\!\rightarrow\!{U}_1}\xspace}
\newcommand\yty{\ensuremath{{U}_2\!\rightarrow\!{U}_2}\xspace}
\newcommand\yts{\ensuremath{{U}_2\!\rightarrow\!{S}}\xspace}
\newcommand\str{\ensuremath{{S}\!\rightarrow\!{R}}\xspace}
\newcommand\stx{\ensuremath{{S}\!\rightarrow\!{U}_1}\xspace}
\newcommand\sty{\ensuremath{{S}\!\rightarrow\!{U}_2}\xspace}
\newcommand\sts{\ensuremath{{S}\!\rightarrow\!{S}}\xspace}

\newcommand\alphaeq{\stackrel{\mathclap{\normalfont\mbox{\scriptsize{(a)}}}}{=}}
\newcommand\betaeq{\stackrel{\mathclap{\normalfont\mbox{\scriptsize{(b)}}}}{=}}

\newcommand\alphadisequal{\stackrel{\mathclap{\normalfont\mbox{\scriptsize{(a)}}}}{\geq}}

\newcommand\betadiseq{\stackrel{\mathclap{\normalfont\mbox{\scriptsize{(b)}}}}{>}}
\newcommand\gammadiseq{\stackrel{\mathclap{\normalfont\mbox{\scriptsize{(c)}}}}{>}}

\newcommand*{\figuretitle}[1]{%
    {\centering  
    \textbf{#1}
    \par\smallskip}}

\newtheorem{definition}{Definition}

\newtheorem{lemma}{Lemma}
\newtheorem{proposition}{Proposition}

\newcounter{assumption}
\newenvironment{assumption}[1][]{\refstepcounter{assumption}
   \noindent \textbf{Assumption \theassumption} (\textit{#1}).}{}
\newcounter{condition}
\newenvironment{condition}[1][]{\refstepcounter{condition}
   \noindent \textbf{Condition \thecondition} (\textit{#1}).}{}
\usepackage{thmtools,thm-restate}
\usepackage{etoolbox} 

\usepackage{algorithm}
\usepackage[noend]{algpseudocode}
\algrenewcommand\algorithmicloop{}

\usepackage{setspace}

\usepackage{pifont}
\newcommand{\V}{{\color{green}\ding{51}}}
\newcommand{\X}{{\color{red}\ding{55}}}

\usepackage[top=0.875in, bottom=1.0in, left=1in, right=1in, columnsep=0.25in]{geometry}
\begin{document}

\title{
    \Large \textbf{
    A scalable estimator of higher-order information \\ in complex dynamical systems}%
    }
\author{Alberto Liardi}
\email{a.liardi@imperial.ac.uk}
\affiliation{Department of Computing, Imperial College London, UK}

\author{George Blackburne}
\affiliation{Department of Experimental Psychology, University College London, UK}
\affiliation{Department of Computing, Imperial College London, UK}

\author{Hardik Rajpal}
\affiliation{\mbox{Center for Complexity Science, Department of Mathematics, Imperial College London, UK}}

\author{Fernando E. Rosas}
\affiliation{\mbox{Center for Complexity Science, Department of Mathematics, Imperial College London, UK}}
\affiliation{Sussex AI and Sussex Centre for Consciousness Science, University of Sussex, Brighton, UK}

\author{Pedro A.M. Mediano}
\affiliation{Department of Computing, Imperial College London, UK}
\affiliation{Department of Experimental Psychology, University College London, UK}

\newbool{inappendix}
\boolfalse{inappendix}

\begin{abstract}
  \noindent
  Our understanding of complex systems rests on our ability to characterise how they perform distributed computation and integrate information. Advances in information theory have introduced several quantities to describe complex information structures, where collective patterns of coordination emerge from higher-order (i.e.\ beyond-pairwise) interdependencies. 
  Unfortunately, the use of these approaches to study large complex systems is severely hindered by the poor scalability of existing techniques. Moreover, there are relatively few measures specifically designed for multivariate time series data. 
  Here we introduce a novel measure of information about macroscopic structures, termed \Minfo, which quantifies the higher-order integration of information in complex dynamical systems. We show that \Minfo can be calculated via a convex optimisation problem, and we derive a robust and efficient algorithm that scales gracefully with system size. Our analyses show that \Minfo is resilient to noise, indexes critical behaviour in artificial neuronal populations, and reflects states of consciousness and task performance in real-world macaque and mouse neuroimaging data. Furthermore, \Minfo can be incorporated into existing information decomposition frameworks to reveal a comprehensive taxonomy of information dynamics. Taken together, these results help us unravel collective computation in large complex systems.
\end{abstract}

\maketitle

\section{Introduction}

How can the collective interactions within a complex system best be measured and understood? This question lies at the core of our efforts to uncover the organising principles of systems as diverse as biological networks, social dynamics, and physical processes~\cite{jensen2022complexity}. 
A defining feature of such systems is the emergence of collective, non-linear dynamics from interactions that go beyond simple pairwise relationships. These \emph{higher-order behaviours}~\cite{rosas2022disentangling, battiston2021physics} give rise to macroscopic properties that are not reducible to individual components, and that often prove more functionally informative and causally predictive of the system's behaviour than lower-level descriptions~\cite{hoel2013quantifying, hoel2025causal, barnett2021dynamical}.

A variety of mathematical frameworks have been developed to characterise such higher-order structures, originating in fields such as network theory~\cite{schneidman2006weak, ganmor2011sparse, yu2011higher}, algebraic topology~\cite{giusti2015clique, sizemore2018cliques, petri2014homological}, and information theory~\cite{williams2010nonnegative,timme2016high,mediano2022greater}. Among these, information-theoretic approaches provide a principled and computationally efficient way to quantify statistical dependencies and their organisation across spatio-temporal scales~\cite{rosas_reconciling_2020, golan2022information, mediano2022greater, luppi2024information}.
Within this framework, higher-order dependencies have been linked to diverse collective phenomena, including the adaptive, resilient, and intelligent behaviour in neural systems~\cite{tononi1998complexity, buzsaki2006rhythms, friston2010free, bassett2011understanding, lecun2015deep, hassabis2017neuroscience, raghu2020survey}, the non-linear critical dynamics near phase transitions in physical models~\cite{matsuda1996mutual, brown2022information, mediano2022integrated}, and the collective behaviour in ecological systems~\cite{rajpal2025information}, economic networks~\cite{rajpal2023synergistic}, and artificial neural networks~\cite{proca2024synergistic, sorokina2008detecting, tsang2017detecting, ehrlich2023a, gutknecht2025shannon, tolle2024evolving}.

However, despite the vast proliferation of measures of higher-order information (see Sec.~\ref{sec:rel_work}), achieving a comprehensive yet computationally tractable quantification of the collective structure of a complex system remains an open problem \cite{amari2001information}. 
Accordingly, the main challenge in %
this area is to devise a measure that:
    (1) captures all higher-order interdependencies between multivariate variables,
    (2) scales efficiently with the system size, 
    (3) is tailored to dynamical systems composed of parts evolving jointly over time, and
    (4) identifies the higher-order contributions explicitly and independently, rather than expressing them only as part of a balance between higher- and lower-order dependencies. %
Existing approaches meet some of these criteria: for instance, conditions (1) and (2) are resolved by the methods in Refs.~\cite{chechik2001group, balduzzi2008integrated, rosas2019quantifying}, while (1), (3) and (4) are satisfied by some information decomposition frameworks~\cite{mediano2025toward, rosas2018information, faes2025partial}.
However, to the best of our knowledge, no existing measure successfully addresses all four challenges simultaneously. 

To tackle this issue, here we introduce a novel estimator of higher-order information in complex dynamical systems, which we term \Minfo. 
Our approach builds on the notion of \emph{union information} 
introduced by Bertschinger \textit{et al.} and Griffith and Koch~\cite{bertschinger2014quantifying, griffith2014quantifying}, which we extend to more general scenarios comprising multiple inputs and outputs. 

We begin by introducing the mathematical foundations of \Minfo %
and propose a practical implementation for its computation that is robust and easily scalable to large systems.
We show that our multivariate formulation of \Minfo has a natural application to dynamical systems, validating this property first on synthetic dynamical models, and then on empirical ECoG data from macaque cortex and neuropixel recordings from mice performing a visual discrimination task.

Our main original contributions include theoretical, computational, and empirical insights. Specifically: 
\begin{enumerate}
    \item We introduce the \Minfo, a novel multivariate information-theoretic measure of higher-order dependencies in complex systems;
    \item We show that \Minfo can be efficiently calculated in Gaussian distributions via a convex optimisation problem;
    \item We illustrate the behaviour of \Minfo on synthetic models, and demonstrate its potential for empirical applications using real-world brain activity data;
    \item We integrate \Minfo with existing information decomposition frameworks, providing a fine taxonomy of information flow in dynamical systems.
\end{enumerate}

\section{A measure of higher-order information} \label{sec:minfo}

\subsection{Preliminaries: Quantifying union information}

Consider a scenario where $d_X$ random variables $\bX=\{X_1, ..., X_{d_X}\}$, referred to as \textit{sources}, are used to obtain information about another variable $Y$, referred to as \textit{target}. Let's denote as $P(X_1,\dots,X_{d_X},Y)$ the joint distribution of sources and target, and as $P(X_i, Y)$ the pairwise marginal distribution between $X_i$ and $Y$. We refer to \textit{lower-order} information as the portion of the dependency structure between sources and target that is captured by pairwise marginals. %
Conversely, \textit{higher-order} information encompasses all dependencies that cannot be reduced to pairwise relationships, such as those involving three or more variables jointly.

A well-known metric to assess lower-order information is the \emph{union information}, which is defined as
    \begin{align} \label{eq:broja}
        I^\cup(\bX; Y) \coloneqq \min_{Q\in\Delta_P} I_Q(\bX; Y) ~ ,
    \end{align}
where $\Delta_P \defeq \{Q\in\Delta : Q(X_i, Y) = P(X_i, Y) ~ \forall i\}$ contains the probability distributions on $(\bX,Y)$ with same pairwise marginals as $P$, $\Delta$ is the full set of all probability distributions on $(\bX,Y)$, and $I_Q$ denotes the mutual information calculated on the distribution $Q$. 
Analogously to the notion of the union operation in set theory, the union information %
can be conceived as capturing the ``union'' of all the pieces of information that each source $X_i$ has about the target $Y$~\cite{bertschinger2014quantifying,griffith2014quantifying}.

Both Bertschinger \textit{et al.}~\cite{bertschinger2014quantifying} and Griffith and Koch~\cite{griffith2014quantifying} introduced $I^\cup$ with the aim of establishing a decomposition of the mutual information $I(\bX; Y)$ (known as BROJA-PID~\cite{james2018dit}, but also sometimes referred to as $\sim$-PID~\cite{venkatesh2023gaussian}) that is grounded in operational principles from game theory~\cite{bertschinger2014quantifying} and network information theory~\cite{tian2025broadcast}.
That said, the union information is an interesting measure of lower-order information in its own right~\cite{kolchinsky2022novel}.
Computationally, calculating the minimum in Eq.~\eqref{eq:broja} for arbitrary distributions is highly non-trivial, with multiple implementations available~\cite{james2018dit,makkeh2018broja,wollstadt2018idtxl,pakman2021estimating}. Conveniently, calculation in Gaussian systems is tractable, and efficient gradient-based methods are available~\cite{venkatesh2023gaussian}.

\subsection{Defining union information in multivariate systems} \label{sec:results_math}

Let's now investigate how to generalise the union information to make it applicable for analysing multivariate dynamical systems -- such as the joint activity of multiple neuronal populations. %
For this, let's consider $\bX=\{X_1, ..., X_{d_X}\}$ source variables and $\bY=\{Y_1, ..., Y_{d_Y}\}$ target variables jointly following a probability distribution $P\in\Delta$.
Our approach is to specify desiderata for the properties of the lower- and higher-order information shared between sources and targets, and then to determine the functional forms that satisfy these criteria. 
Following this rationale, we first adopt the inclusion-exclusion principle \cite{williams2010nonnegative, kolchinsky2022novel}, assuming that the higher-order information can be obtained as the difference between the mutual information between sources and targets and the lower-order information. Then, we propose the following additional properties:
    
    \begin{assumption}[Pairwise dependence]
        \label{ass:low}
        Lower-order information should only depend on the pairwise marginal distributions $P(X_i, Y_j)$.
    \end{assumption}
    
    \begin{assumption}[Non-negativity]
        \label{ass:nonneg}
        The amount of lower- and higher-order information cannot be negative.
    \end{assumption}
    
Assumption~\ref{ass:low} formalises the intuition that lower-order information consists of the ``union'' of all the pieces of information that any $X_i$ holds about any $Y_j$. This requirement is theoretically supported from considerations on probability mass exclusion~\cite{finn2018pointwise}, and follows from operational considerations from game theory~\cite{bertschinger2014quantifying} and information network theory~\cite{tian2025broadcast}. This condition is also satisfied by many (though not all) commonly used measures of information decomposition~\cite{williams2010nonnegative,harder2013bivariate,bertschinger2014quantifying,barrett2015exploration, liardi2026mathematical}.
On the other hand, Assumption~\ref{ass:nonneg} is a desirable property of information decompositions in general, aiding their interpretability by allowing an understanding both in set-theoretic terms~\cite{kolchinsky2022novel, down2025algebraic}, and in the classical communication-theoretic sense\footnote{Nonetheless, several measures do not satisfy this condition~\cite{ince2017measuring,rosas2020operational} but are also theoretically justified~\cite{finn2018pointwise}.}~\cite{james2018uniquekey}. 
Unfortunately, we will show that Assumptions~\ref{ass:low}-\ref{ass:nonneg} are not sufficient to uniquely define a functional expression for these quantities. Hence, we introduce the following additional condition:

\begin{condition}[Existence]
    \label{cond:existence}
    For a given $P\in\Delta$, there exists a $Q^*$ with  ${Q^*(X_i,Y_j)=P(X_i, Y_j)}$ that has no higher-order information.
\end{condition}

Condition~\ref{cond:existence} formalises the intuitive idea that, for a system $P$, the higher-order dependencies can be arbitrarily reduced by acting solely on the joint distribution, while leaving the pairwise relationships unaffected.
In other words, it asserts that for any $P$ there exists a distribution $Q^*$ in $\Delta_P$ whose statistical structure can be fully captured by pairwise relations alone, i.e.\ by the lower-order dependencies.
However, formally establishing that this condition holds for an arbitrary system is non-trivial~\cite{bertschinger2014quantifying}. %

Building on these considerations, we now introduce two information-theoretic quantities as candidate measures of lower- and higher-order information. \\
\begin{definition} \label{def:wm}
    Given two random vectors $\bX$ and $\bY$, the \W- and \Minfo are defined as:
    \begin{align} \label{eq:Winfo_def_main}
        \W(\bX; \bY) &:= \min_{Q\in\Delta_P} I_Q(\bX;\bY) \,, 
        \\ \label{eq:Minfo_def}
        \M(\bX; \bY) &:= I(\bX; \bY) - \W(\bX; \bY) \,,
    \end{align}
    with
    $\Delta_P \defeq \{Q\in\Delta : Q(X_i, Y_j) = P(X_i, Y_j) ~ \forall i,j\} \,$\ifbool{inappendix}{}{%
      \footnote{With a slight abuse of notation, we use the same symbol $\Delta_P$ as in Eq.~\eqref{eq:broja}. We consider this justified, as both definitions are the same (and equal to that of Refs.~\cite{bertschinger2014quantifying,griffith2014quantifying}) when $d_Y=1$.}%
    }.
\end{definition}

Eq.~\eqref{eq:Winfo_def_main} represents a multivariate extension of the union information of Eq.~\eqref{eq:broja}, which we thus informally call ``double union'' information and denote as \W (corresponding to two union signs).
With these definitions at hand, our first result is to show that, under Assumptions~\ref{ass:low}-\ref{ass:nonneg} and Condition~\ref{cond:existence}, $\mathcal{W}$ and $\mathcal{M}$ are precisely the lower- and higher-order information. \\

\begin{restatable}{proposition}{propWdef} \label{prop:w}
    Under Assumptions~\ref{ass:low}-\ref{ass:nonneg} and Condition~\ref{cond:existence}, the \W- and \Minfo are the unique quantities that measure the lower- and higher-order information of a system, respectively.    
    If Condition~\ref{cond:existence} is not satisfied, then \W and \M become upper- and lower-bounds to the lower- and higher-order information.
\end{restatable}

Hence, together, the expressions of \W- and \Minfo above provide a principled way of decomposing the information shared between two random vectors into lower-order (i.e.\ pairwise) and higher-order (i.e.\ beyond-pairwise) components. Moreover, this approach is particularly convenient for the study of time series and dynamical systems, for which there is a natural interpretation: the source $\bX$ corresponds to the past state of the system, and $\bY$ to the future state.

\subsection{Calculating union and higher-order information in Gaussian systems}
\label{sec:results_implem}

While the minimisation in Eq.~\eqref{eq:Winfo_def_main} is in general challenging, restricting the analysis to jointly Gaussian $\bX,\bY$ allows us to prove our main theorem, which enables us to design an efficient algorithm for computing \W: \\

\begin{restatable}{theorem}{theoConvex} \label{theo:convex}
The minimisation of Eq.~\eqref{eq:Winfo_def_main} for the computation of \Winfo can be solved via convex optimisation, and admits a unique solution $Q^*\in\Delta_P$.
\end{restatable}

Thanks to Theorem~\ref{theo:convex}, the \Winfo can be efficiently calculated using standard tools from convex optimization~\cite{boyd2004convex, marshall1979inequalities}. 
In the following, we provide an efficient unconstrained parametrisation for the feasible set $\Delta_P$ and perform the optimisation for \Winfo. A complete technical description is provided in the Appendix.

Consider the joint Gaussian distribution $P(\bX, \bY)$ with zero mean and full-rank covariance $\Sigma$, and let $L$ be its Cholesky decomposition, i.e.\ ${\Sigma = L L^\top}$. Using a block matrix representation, we can write:
\begin{equation} \label{eq:cholblock_double}
    \Sigma 
     = \begin{pmatrix} \Sigma_{XX} & \Sigma_{XY} \\ \Sigma_{YX} & \Sigma_{YY} \end{pmatrix}  = L L^\top 
    = \begin{pmatrix} L_{XX} & 0 \\ L_{YX} & L_{YY} \end{pmatrix}
      \begin{pmatrix} L_{XX}^\top & L_{YX}^\top \\ 0 & L_{YY}^\top \end{pmatrix} 
\end{equation}
where $\Sigma_{XY} = \Sigma_{YX}^\top$, $L_{XX}$ and $L_{YY}$ are lower-triangular diagonal-positive matrices, and $L_{XY}$ is any real matrix.

The constraint defining the optimisation set $\Delta_P$ requires that $\Sigma^Q_{YX} = \Sigma_{YX}$, with $Q$ denoting the covariance of the optimised distribution.
Enforcing this condition in Eq.~\eqref{eq:cholblock_double}, we obtain that the covariance matrix of a distribution $Q \in \Delta_P$ can be written as
\begin{align} \label{eq:chol_delta_P}
    \Sigma^Q = L^Q L^{Q,\top} ~ ,\quad \text{with}\quad L^Q = \begin{pmatrix} L^Q_{XX} & 0 \\ \Sigma_{YX} (L^{Q,\top}_{XX})^{-1  } & L^Q_{YY} \end{pmatrix} \,,
\end{align}
where $L_{XX}^Q$ and $L_{YY}^Q$ are any lower triangular matrices\footnote{The constraints in $\Delta_P$ also require that $\mathrm{diag}(\Sigma^P)=\mathrm{diag}(\Sigma^Q)$. We address this in the Appendix.}.
Finally, we can parametrise $L_{XX}^Q$ and $L_{YY}^Q$ using a vector of real numbers, employing techniques from Ref.~\cite{leger2023parametrization}.

Thus, following this procedure, the original constrained optimisation over the space of covariance matrices is reformulated as an unconstrained problem over the real numbers. Since all mappings involved are differentiable, the optimisation can be performed using standard gradient-based methods. 
In this work, all optimisations were performed using Adam~\cite{kingma2014adam}.

We note that this approach assumes the optimising solution is jointly Gaussian, therefore yielding, in principle, only an upper bound on the true minimising solution. However, as shown in the Appendix, the joint Gaussian solution is exact for certain systems and constitutes a tight upper bound in more general cases.

\section{Results} \label{sec:results}

\subsection{Validation on synthetic data} \label{sec:results_synthetic}

\textbf{Experiment 1: Toy systems.}
We start by calculating $\W(\bZ_t; \bZ_{t+1})$ and $\M(\bZ_t; \bZ_{t+1})$ for bivariate systems $\bZ_t=(Z_{t}^1, Z_{t}^2)$ for which we have clear intuitions. We consider three scenarios: a ``COPY'' system, where $Z_{t}^i$ is strongly coupled to $Z_{t+1}^i$; a ``transfer'' system, where $Z_{t}^i$ is strongly coupled to $Z_{t+1}^j$ ($i\neq j$); and a Gaussian analogue to the XOR gate, which we expect to show higher-order information (Fig.~\figsubref{fig:synthetic_res}{a}; see Appendix for a full description).

Results show that the decomposition of the input-output information is effectively decomposed into $\M$ and $\W$, with the \Minfo only assigning higher-order components to the XOR gate. We remark that the small (unexpected) contributions of $\M$- and $\W$-information in COPY and XOR, respectively, are caused by numerical artefacts due to the input covariances being close to singular. 

\begin{figure}[t]
  \centering
  \includegraphics[width=\textwidth]{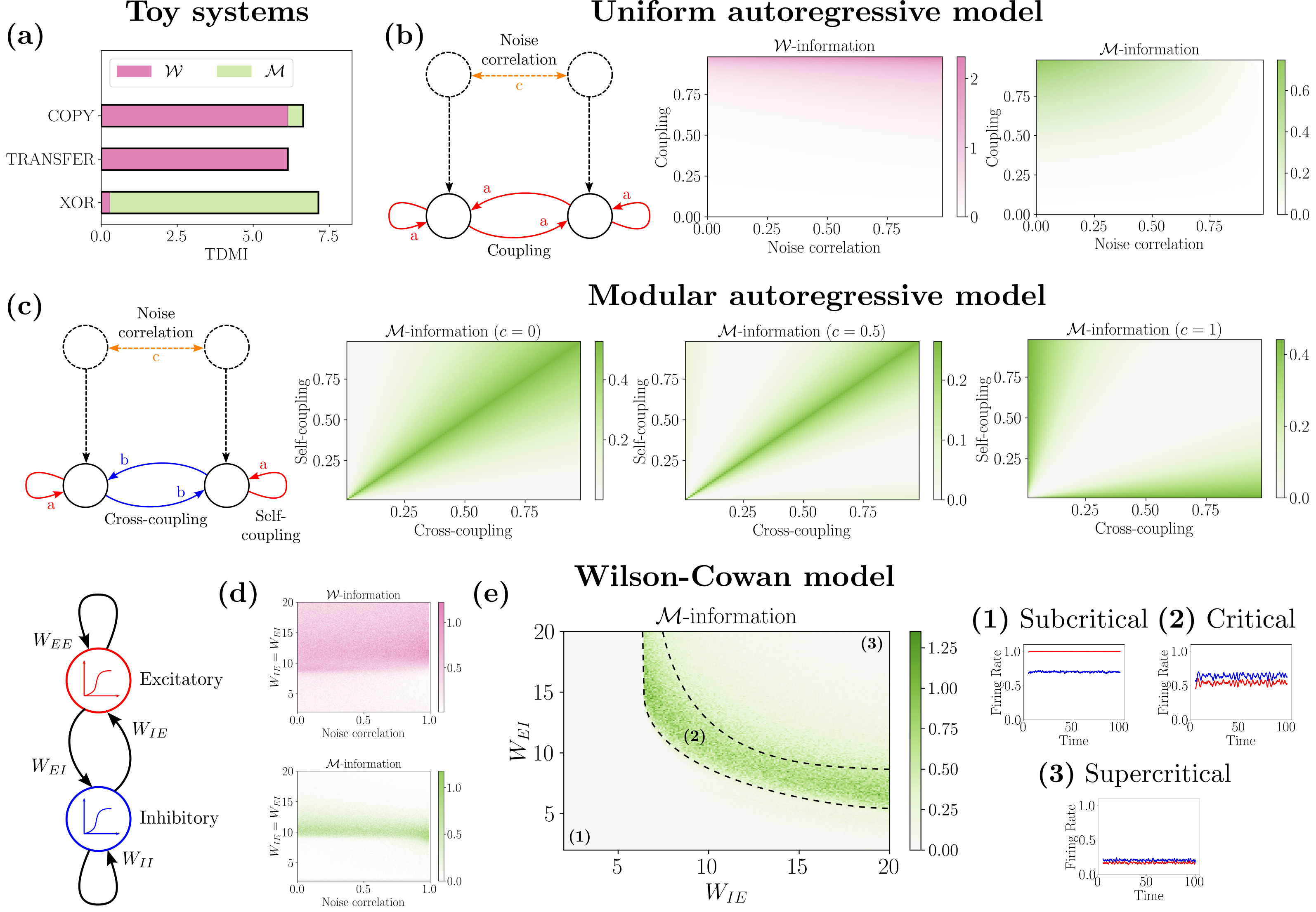}
  \caption{\textbf{\Minfo captures higher-order structures in linear and non-linear systems.}
  \textbf{(a)} Computation of \W- and \Minfo on COPY, Transfer, and XOR Gaussian systems shows that, as expected, only the XOR contains higher-order information.
  \textbf{(b)} The uniform VAR model of Eq.~\eqref{eq:VAR_eq1} exhibits the highest \Winfo for high coupling and noise correlation $(a\to1,\,c\to1)$, and the highest \Minfo for large coupling and low noise correlation $(a\to1,\,c\to0)$.
  \textbf{(c)} For a VAR with modular interaction terms (Eq.~\eqref{eq:VAR_eq2}), \Minfo is maximum for equal coefficients when the noise correlation is low $(a\approx b,\,c\lesssim0.5)$, and for uneven coefficients when the noise correlation is high.
  \textbf{(d)} In the Wilson-Cowan model, the estimation of \W and \Winfo is robust to the injection of noise correlation in the system.
  \textbf{(e)} \Minfo identifies 3 distinct regions in the phase space of the Wilson-Cowan dynamics, being highest in the critical regime.}  
  \label{fig:synthetic_res}
\end{figure}

\textbf{Experiment 2: Uniform autoregressive process.}
We continue with the validation of the method by considering different Vector Autoregressive (VAR) systems on which we calculate \W- and \Minfo. 
Specifically, given a VAR model
\begin{align} \label{eq:var_eq_gen}
    \bZ_{t+1} = A\bZ_t + \boldsymbol{\varepsilon}_t \,, 
\end{align}
with $\bZ=(Z^1, ..., Z^n)$, $\boldsymbol{\varepsilon}_t\sim\mathcal{N}(0,V)$, setting different parameters $A,V$ allows us to study how different dynamics are captured by \W. Here we show the case $n=2$, but similar results hold for larger systems (see Appendix).  
As above, we calculate $\W(\bZ_t; \bZ_{t+1})$ and $\M(\bZ_t; \bZ_{t+1})$, effectively analysing how information is integrated over time by the dynamics of the system. 
We consider a VAR model with parameters
\begin{equation} \label{eq:VAR_eq1}
    A = \frac{1}{2}
    \begin{pmatrix}
    a & a \\ a & a     
    \end{pmatrix}
    \quad \text{and} \quad
    V = 
    \begin{pmatrix}
    1 & c \\ c & 1     
    \end{pmatrix} \,,
\end{equation}
with $a,c\in(0,1)$,
so that the spectral radius of $A$ is 1 when $a=1$. 

Calculating $\M$- and $\W$-information for various values of $(a,c)$ yields interpretable results (Fig.~\figsubref{fig:synthetic_res}{b}): \W grows as $c$ approaches 1 and the system is more redundant; and \M is highest when all system elements are tightly coupled ($a\to1$) and the noise is uncorrelated ($c\to0$), effectively capturing dynamical cooperation not mediated by noise. \\

\textbf{Experiment 3: Modular autoregressive process.} To further study the interplay between dynamics and \Minfo, we consider a VAR of the form:
\begin{equation} \label{eq:VAR_eq2}
    A = 
    \begin{pmatrix}
    a & b \\ b & a     
    \end{pmatrix}
    \quad \text{and} \quad
    V = 
    \begin{pmatrix}
    1 & c \\ c & 1     
    \end{pmatrix} \,,
\end{equation}
with $a,b,c\in(0,1)$. 
Fixing the parameter $c$, we perform an analogous analysis as before by varying $a,b$ and calculating $\M$ and $\W$. To facilitate comparison across the different systems, for each pair of $(a,b)$ we rescale the spectral radius of $A$ so that the system always possesses the same mutual information $I(\bZ_t; \bZ_{t+1})$. 
This ensures that differences in \M- and \W-information are only due to the internal dynamical structure of the model, and not to the total information content~\cite{liardi2024null}. 

Results show that $\M$ is larger for low $c$ only when both self- and cross-coupling are high, as neither alone explains the system's dynamics (Fig.~\figsubref{fig:synthetic_res}{c}). Interestingly, this pattern is reversed for high $c$, plausibly because the asymmetry in the interactions (either high self and low cross-coupling, or vice versa) overcomes the effect of noise that redundantly correlates the sources. \\

\textbf{Comparison with existing information-theoretic quantities.}
Before further validating our measure on more elaborate synthetic systems, we examine how \W and \M relate to other existing information-theoretic measures of lower- and higher-order dependencies. In particular, we focus on the behaviour of the O-information \cite{rosas2019quantifying}, its extension to time-evolving processes (O-information rate, OIR) \cite{faes2022new}, and the Redundancy-Synergy Index (RSI) \cite{chechik2001group}. We refer to the Appendix for their mathematical definition. 
We compute these quantities on the uniform and modular autoregressive models introduced in the previous paragraphs, as well as on an additional lagged system with multiple past time steps (see Appendix). For comparison, all quantities are computed on the full covariance matrix of the system. The O-information is thus evaluated over the joint distribution of past and future states, whereas OIR, RSI, \W, and \Minfo are calculated between past and future states.

Across most regimes, these measures exhibit qualitatively similar trends. Nevertheless, while \W and \M respond to changes in both noise correlation and self- and cross-coupling, O-information and RSI sometimes fail to differentiate between values of $c$ (with an average relative change across $c$ an order of magnitude smaller than that of OIR and \W-\M), and OIR struggles to detect variations in $a$ (average relative change smaller than $2\%$). Moreover, these measures capture only the balance between higher- and lower-order information, whereas \W and \M provide explicit estimates of both contributions. 
Finally, we remark that O-information is not ideally suited for the study of dynamical processes, as it does not distinguish between correlations arising from the steady-state distribution and those emerging from temporal evolution. In contrast, the directed metrics \W and \M, OIR, and RSI can account for temporal structure by defining $\bX$ as the system's past state and $\bY$ as its future state, thereby incorporating dynamical dependencies into the analysis.
The key findings of this investigation are summarised in Table~\ref{tab:measures_comparison}, with the detailed results and additional discussion reported in the Appendix.

\begin{table}[t]
\centering
\setlength{\tabcolsep}{7pt} %
\renewcommand{\arraystretch}{1.1} %
\caption{\textbf{\W- and \Minfo satisfy more useful properties than alternative measures for analysing dynamical systems.} Summary of comparison between O-information, O-information rate (OIR), Redundancy-Synergy Index (RSI) and \W and \M on Uniform, Modular, and Lagged autoregressive models. See Appendix for details.}
\label{tab:measures_comparison}
\begin{tabular}{ccccc}
\toprule
Measure & \makecell{Sensitive to\\coupling strength} 
        & \makecell{Sensitive to\\noise correlation} 
        & \makecell{Distinguishes steady-state\\vs dynamic dependencies} 
        & \makecell{Separates higher-\\and lower-order} \\
\midrule
O-information & \V & \X & \X & \X \\
OIR          & \X  & \V & \V & \X \\
RSI          & \V & \X & \V & \X \\
\W, \M       & \V & \V & \V & \V \\
\bottomrule
\end{tabular}
\end{table}

Note that these conclusions may not (and were not intended to) generalise to other complex systems, as the behaviour of each measure depends on the specific system and parameters investigated. Nonetheless, these findings suggest that \W and \M can offer greater sensitivity to dynamical parameters while disentangling lower- and higher-order interdependencies, hence addressing the desiderata (3) and (4). \\

\begin{figure}[b]
    \centering
    \includegraphics[width=\linewidth]{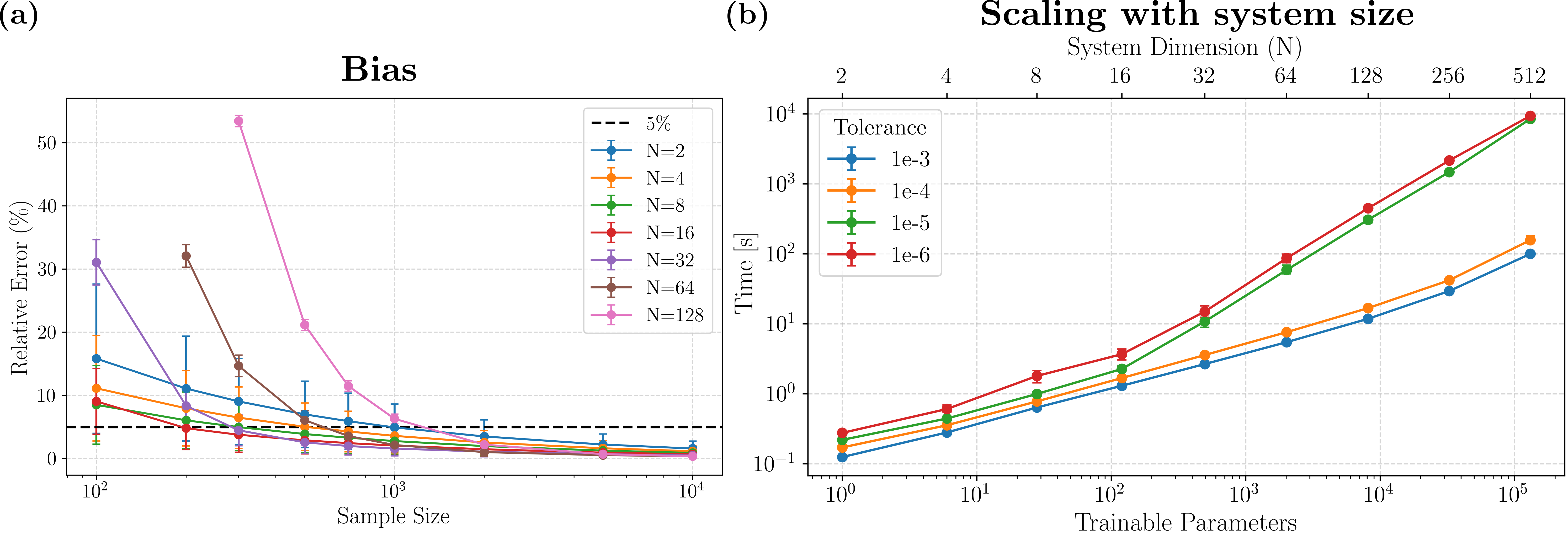}
    \caption{\textbf{The estimation of \Winfo is accurate to within a few percentage points, and its computation time scales efficiently with system size.} \textbf{(a)} Relative percentage accuracy of \Winfo for different sample sizes and system dimensions. \textbf{(b)} Runtime as a function of trainable parameters of the system and system dimensions. Tolerances specify the convergence thresholds used in the optimisation. Error bars indicate the SEM.}
    \label{fig:acc_time_tests}
\end{figure}

\textbf{Experiment 4: Non-linear neural dynamics.}
Although the algorithm presented above was derived for Gaussian variables, we show that, in practice, our estimation of higher-order interdependencies can be useful in more general (non-linear) cases.  
As an example, we consider a Wilson-Cowan (WC) model of interacting excitatory and inhibitory neuronal populations with non-linear gain functions, and analyse its information dynamics in different dynamical regimes.

We estimate the \Minfo by fitting a Gaussian copula to the simulated steady-state distributions of neuronal activity, and employ the optimisation technique presented in Sec.~\ref{sec:results_implem} across a range of excitatory-inhibitory $(W_{EI})$ and inhibitory-excitatory $(W_{IE})$ weights, as well as different values of the noise correlation $c$.  We observe that \Minfo delineates a specific region of higher-order dependencies, which interestingly corresponds to the critical regime of sustained periodic oscillations (Fig.~\figsubref{fig:synthetic_res}{e}). Furthermore, sweeping over various values of $c$ shows that the \Minfo is robust to noise (Fig.~\figsubref{fig:synthetic_res}{d}). In sum, we show that \Minfo is sensitive to biophysically meaningful modes of non-linear neural activity.

\begin{figure}[t]
    \centering
    \includegraphics[width=1\linewidth]{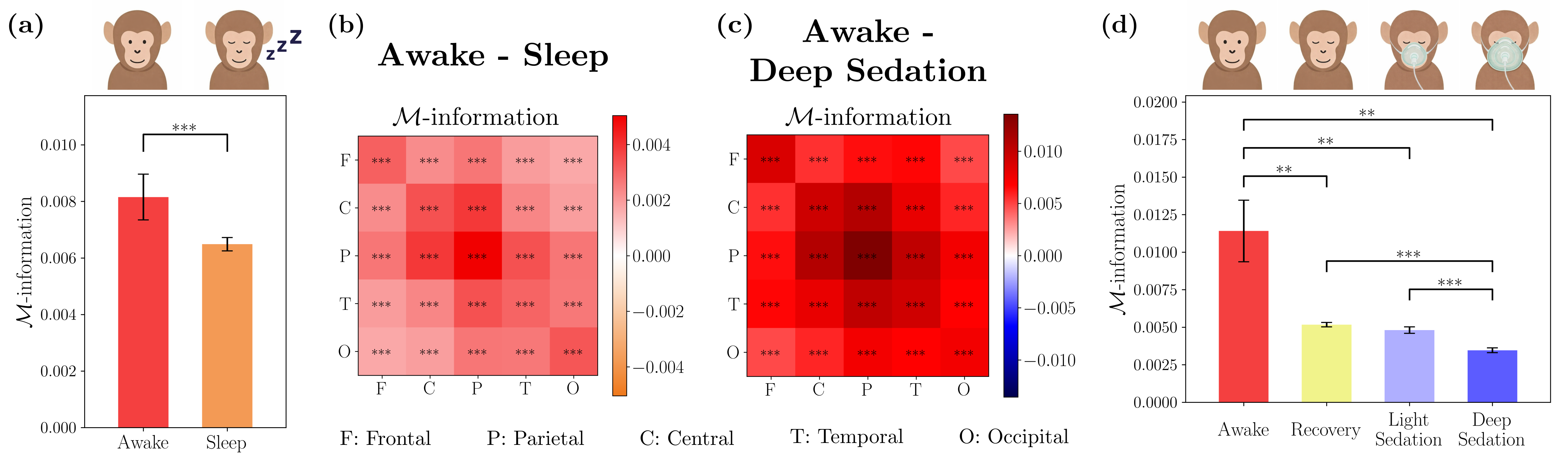}
    \caption{\textbf{\Minfo tracks level of consciousness in monkeys during awake, sleep, and sedated states.} \textbf{(a)} Average of \Minfo across all pairs of electrodes in awake and sleep conditions. \textbf{(b)} Difference of \Minfo between awake and sleep conditions averaged across all pairs belonging to the five major cortical areas (Frontal, Parietal, Central, Temporal and Occipital). \textbf{(c)} same as (b) but for awake and deep anaesthesia. \textbf{(d)} same as (a) but for awake, recovery, light and deep anaesthesia.
    All the \Minfo values above are normalised w.r.t.\ the mutual information. (P-values Holm-Bonferroni corrected calculated with a one-sample t-test against the zero-mean null hypothesis. *: $p<0.05$; **: $p<0.01$; ***: $p<0.001$).}
    \label{fig:monkey}
\end{figure}

\subsection{Bias correction and scalability} \label{sec:results_synthetic_stress}

We now study how the accuracy and runtime of our method scale with system size.
We test the former by sampling a random covariance matrix from a Wishart distribution, generating data of arbitrary length $L$, and estimating the covariance of the system again. Finally, we calculate \Winfo and apply the bias correction suggested by Venkatesh \textit{et al.}~\cite{venkatesh2023gaussian} to account for the finite sample sizes (Fig.~\figsubref{fig:acc_time_tests}{a}). 
For smaller systems, the estimation error swiftly falls beneath the 5\% threshold, and even for larger systems, the bias remains within an acceptable range with approximately 1,000 samples. This makes it feasible to obtain accurate results on data such as fMRI~\cite{van2012human}, which typically consists of a few hundred time points and is known to pose challenges due to short time series and poor temporal resolution.
Finally, we notice that \Winfo calculated on a sample size of 10,000 provides an error below 1\% for all system sizes studied here.

We then test the computational speed of our method on randomly sampled covariance matrices from Wishart distributions, showing that the optimisation scales efficiently with the system size and sub-linearly with the trainable parameters of the model (Fig.~\figsubref{fig:acc_time_tests}{b}).
Overall, we find that the estimation of \Winfo -- and therefore of \Minfo -- is accurate, robust to different sampling sizes, and computationally efficient.

\subsection{Empirical applications}
\label{sec:results_empirical}

\textbf{Loss of consciousness in the macaque brain.}
We now present an application of \W and \M to real-world electrophysiology data.  
Specifically, we apply \W- and \Minfo to ECoG recordings from a macaque during different physiologically- and pharmacologically-induced states of consciousness \cite{yanagawa2013large}. 
The dataset comprises 128-channel ECoG recordings from the right hemisphere of a monkey sitting calmly with its head and arms restrained. The experimental conditions include wakefulness, natural sleep, and light and deep sedation with the general anaesthetic drug propofol. 

We first compute \W- and \Minfo for all electrode pairs and then average across them to obtain global values, finding that \Minfo is consistently higher during wakefulness than during sleep or sedation (\figsubref{fig:monkey}{a}–\figsubref{fig:monkey}{d}). Moreover, we observe that this behaviour is spatially homogeneous, with the five major cortical areas being modulated comparably across states (\figsubref{fig:monkey}{b}-\figsubref{fig:monkey}{c}). Thus, \Minfo reveals that reduced conscious level coincides with reduced higher-order information, aligning with previous results in the literature~\cite{imperatori2021cross, schartner2015complexity} and validating the interpretability of the framework on real-world data. 
Additionally, these results not only reproduce but also extend the findings previously obtained using OIR on the same dataset \cite{faes2022new}. In fact, \Winfo shows that there is no significant difference in the lower-order dependencies across conditions (Appendix). Therefore, \W and \M collectively suggest that the previously reported increase in the balance between higher- and lower-order interdependencies is driven by an increase of the former, rather than by a reduction of the latter.

\textbf{Visual decision-making in the mouse brain.}
We finally apply our measures to a neuroimaging dataset of mice performing a visual discrimination task~\cite{steinmetz2019distributed}. The dataset includes high-frequency local field potential (LFP) recordings from different brain regions while contrasting visual stimuli were presented on either side of a screen. Mice were trained to rotate a wheel to bring the high-contrast image towards the centre of the screen, with trials classified as positive (correct) or negative (incorrect) based on the accuracy of the choice. In certain trials where stimuli with equal contrast were presented on both sides, mice were randomly rewarded, and we refer to these trials as uncertain trials. Finally, as a control for stimulus presence, passive trials were performed where visual stimuli were presented with no instruction or reward.
For each mouse and experimental session, multiple Neuropixel probes were placed across the brain to provide broad coverage of cortical and subcortical structures. 
We employed \Minfo to assess the presence of higher-order structures across trials, brain regions, and throughout the recording session. Since \Winfo provides analogous results but from the lower-order perspective, we report those findings in the Appendix. 

\begin{figure}[t]
    \centering
    \includegraphics[width=\linewidth]{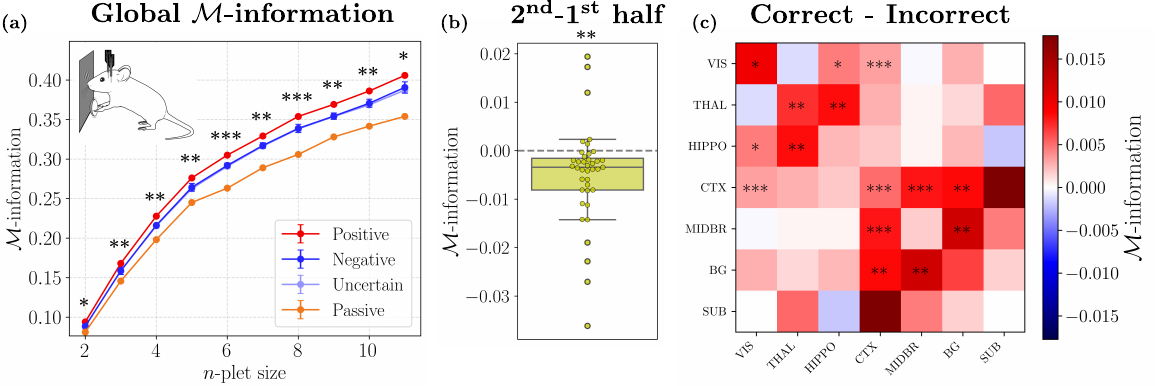}
    \caption{\textbf{\Minfo in mouse neural activity captures perceptual performance in a visual discrimination task} \textbf{(a)} \Minfo of different trials across different $n$-plets sizes. Significance tests were performed between positive and negative trials. Error bars indicate the SEM. \textbf{(b)} Difference of \Minfo between the second and first half of all active trials. \textbf{(c)} Difference of \Minfo between correct and incorrect trials per regional interaction. All the \Minfo values above are normalised w.r.t.\ to the mutual information. (P-values calculated with a one-sample t-test against the zero-mean null hypothesis. *: $p<0.05$; **: $p<0.01$; ***: $p<0.001$).}
    \label{fig:Neuro}
\end{figure}

We first apply \Minfo to $n-$plets of LFP probes, with $n\in[2,11]$ (depending upon the recording session) in each of the four conditions. We observe that population activity has the highest \Minfo when the mouse makes correct perceptual decisions, and the lowest when remaining passive. Stronger effects were obtained for larger $n-$plets, underscoring the importance of multivariate measures for accurately discriminating system behaviour (Fig.~\figsubref{fig:Neuro}{a}). 
By computing \Minfo on the first and second half of the trials in each run for each pair of pixels, we observe a decrease in $\M$ over time (Fig.~\figsubref{fig:Neuro}{b}). This suggests that perceptual performance is initially subserved by distributed cooperation via macroscopic information processing, but once the mouse has accommodated to the task, computation is offloaded to lower-order elements.

To further investigate the biological orchestration of visual perception, we examine potential higher-order interactions within and between specific anatomical areas (Fig.~\figsubref{fig:Neuro}{c}). Comparing correct and incorrect perceptual decisions, we find a profound effect of increased higher-order interactions within the neocortex, as well as between the neocortex and midbrain, basal ganglia, and visual cortex, respectively. We also observe significant increases for hippocampo-visual and thalamo-hippocampal interactions, likely reflecting the successful coordination of sensory variables with contextual memory and task-relevant information. Intriguingly, however, we do not find a significant effect for increased thalamo-visual interactions, which may be a consequence of the functional and cytoarchitectural heterogeneity of thalamic nuclei (see Appendix). Notably, we also observe greater \Minfo within the visual cortex, as well as within the thalamus, indicating they function more as wholes, rather than parts, during perceptual processing. In conclusion, \Minfo reveals novel effects for the informational organisation of mammalian brains encoding behaviourally relevant variables.

\section{Information Decomposition based on \Winfo} \label{sec:broja-phiid}

Having established that \W- and \Minfo are meaningful measures in their own right, we finally show that the methodology developed above can be integrated with existing information-theoretic frameworks.

Part of the original motivation for the works in Refs.~\cite{bertschinger2014quantifying,griffith2014quantifying} was to decompose mutual information 
using the Partial Information Decomposition (PID) framework~\cite{williams2010nonnegative}. PID aims to separate the total information the sources $\bX$ have about the target $Y$ into redundant, unique, and synergistic components. This allows one to disentangle the amount of information shared by all sources, unique to each of them, and that arises only from their joint cooperation, respectively.   
In the case of two sources $X_1,X_2$ and a target $Y$, PID states that marginal and joint mutual information can be written as
\begin{subequations} \label{PID_eqs}
\begin{align}
    I(X_1;Y) &= R + U_1 \label{PID_eqs:a} \\
    I(X_2;Y) &= R + U_2 \label{PID_eqs:b} \\
    I(X_1,X_2;Y) &= R + U_1 + U_2 + S \label{PID_eqs:c} \,,
\end{align}
\end{subequations}
where $R$ represents redundancy, $U_1, U_2$ unique information of $X_1,X_2$, respectively, and $S$ synergy. 
Importantly, these equations do not define a unique functional form for $R$, $U_i$ and $S$, and many different approaches have been proposed to define these contributions~\cite{williams2010nonnegative, bertschinger2013shared, griffith2014quantifying, harder2013bivariate, barrett2015exploration, ince2017measuring, james2018unique, griffith2014intersection, griffith2015quantifying, quax2017quantifying, rosas2020operational, liardi2026mathematical}.

\begin{figure}
    \centering
    \includegraphics[width=1\linewidth]{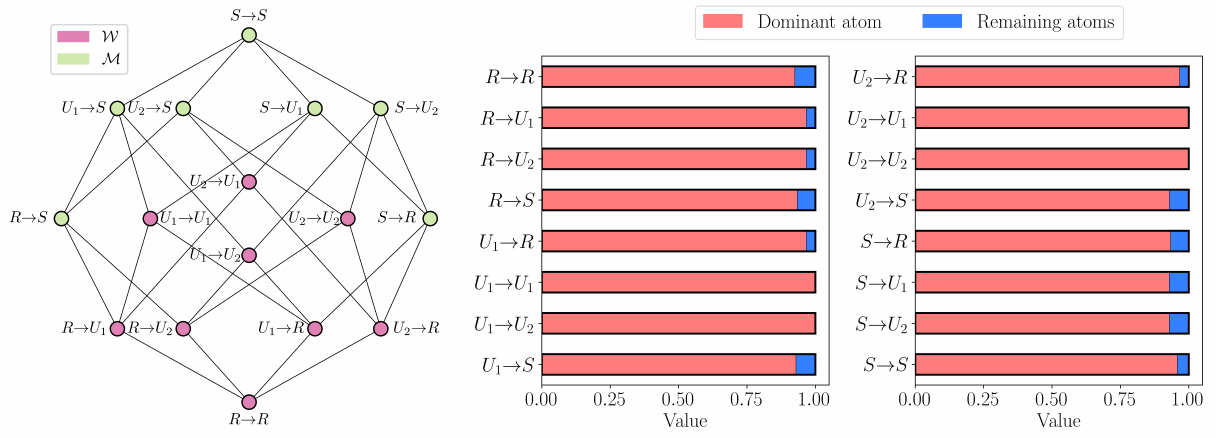}
    \caption{\textbf{BROJA-\phiid successfully characterises information flow in toy systems.} \textbf{(a)} Components of \phiid arranged in a lattice, with terms highlighted according to their role in \W- and \Minfo. \textbf{(b)} BROJA-\phiid results on toy systems which are expected to have one dominant atom (reported on the $y$ axis, red). The sum of all remaining atoms in each case is shown in blue.
    }
    \label{fig:dominant_tests_phiID}
\end{figure}

However, despite its flexibility, the PID formalism does not enable a decomposition of multivariate targets. A generalisation of PID to multiple targets has then been proposed under the name of Integrated Information Decomposition (\phiid) \cite{mediano2025toward}. 
Considering two bivariate random variables $\bX=(X_1,X_2)$ and $\bY=(Y_1,Y_2)$, \phiid aims to decompose the mutual information $I(\bX,\bY)$ into $R$, $U$, and $S$ components both in sources and targets, providing a finer-grained decomposition than PID. For instance, information that is redundant in the sources (i.e.\ between $X_1$ and $X_2$) and uniquely transferred to the target $Y_1$ would be denoted as $\rtx$. Alternatively, information contained in only source $X_2$, which then becomes synergistic in the targets $Y_1$ and $Y_2$, is represented by \xts. A similar reasoning applies to all other possible combinations of $R$, $U$, and $S$, giving birth to sixteen information modes~\cite{mediano2025toward}.

Calculating \phiid requires both a choice of PID and an additional condition. For the former, we adopt the BROJA-PID~\cite{bertschinger2014quantifying}; for the latter, we propose a similar reasoning to that of Bertschinger \textit{et al.}~\cite{bertschinger2014quantifying}, arguing that all the lower-order components should depend solely on redundancy and unique information terms. Thus, they should be included in the \Winfo. This provides the following additional constraint:
\begin{equation} \label{eq:W_phiid}
\begin{split}
    \W =& \rtr + \rtx + \rty +\\ &\xtr + \xtx + \xty +\\ &\ytr + \ytx + \yty \,.
    \end{split}
\end{equation}

Hence, with Eq.~\eqref{eq:W_phiid} and the BROJA-PID at hand, we can perform the full \phiid decomposition of the system. We name this approach BROJA-\phiid.
We report a more detailed mathematical explanation in the Appendix.

To validate our proposed decomposition, we analyse toy models similar to those considered in Sec.~\ref{sec:results}, but designed to maximise a specific \phiid atom (Fig.~\ref{fig:dominant_tests_phiID}). Indeed, we find that in all cases, BROJA-\phiid assigns the dominant contribution to the correct atom, proving its ability to discriminate among different types of integration dynamics. Hence, not only is it possible to assess the total higher-order information via \Minfo, but also to uncover which specific information atoms are responsible for it via BROJA-\phiid.

\section{Related work} \label{sec:rel_work}

A great number of higher-order information measures exist in the literature, which include (among others) the interaction information~\cite{gat1998synergy, brenner2000synergy, schneidman2003synergy, anastassiou2007computational}, the Delta-information~\cite{nirenberg2001retinal, latham2005synergy}, the O-information~\cite{rosas2019quantifying}, the redundancy-synergy index~\cite{chechik2001group}, Varadan's synergy~\cite{varadan2006computational}, information decomposition approaches~\cite{williams2010nonnegative}, and measures aimed at quantifying emergent behaviour~\cite{rosas_reconciling_2020, barnett2023dynamical, rosas2024software}. A separate, prominent line of work quantifying higher-order information in dynamical systems is integrated information theory~\cite{balduzzi2008integrated,oizumi2014phenomenology}, within which multiple measures are available~\cite{mediano2018measuring}, including some based on information decomposition~\cite{mediano2025toward} and information geometry~\cite{oizumi2016unified,langer2020complexity}. 

Among the many higher-order information measures, information decomposition approaches play a fundamental role due to their ability to disentangle distinct modes of information sharing. The seminal work of Williams and Beer introduced the framework of Partial Information Decomposition (PID), in which the mutual information between an arbitrary number of sources and a single target is divided into redundant, unique, and synergistic terms~\cite{williams2010nonnegative}. Following this work, many different PID formulations have been suggested ~\cite{williams2010nonnegative, bertschinger2013shared, griffith2014quantifying, harder2013bivariate, barrett2015exploration, ince2017measuring, james2018unique, griffith2014intersection, griffith2015quantifying, quax2017quantifying, rosas2020operational}, including the BROJA definition proposed contemporaneously by Bertschinger \textit{et al.}~\cite{bertschinger2014quantifying} and Griffith and Koch~\cite{griffith2014quantifying}, which forms the inspiration of this work. 
Ince introduced an entropy-based decomposition~\cite{ince2017partial, varley2023partial}, Varley a Kullback-Leibler divergence decomposition~\cite{varley2024generalized}, and Mediano \textit{et al.} generalised PID by decomposing both information from sources and targets symmetrically~\cite{mediano2025toward}. Unfortunately, however, application of these frameworks to large datasets is limited by the superexponential growth of information modes~\cite{jansma2025fast}.

Various algorithms have been previously introduced to improve the computation of PID quantities, as well as information quantities more broadly.
Makkeh \textit{et al.} proposed a refinement of the original BROJA-PID optimisation for the discrete case via cone programming~\cite{makkeh2018broja}, while Venkatesh \textit{et al.} improved it for Gaussian variables using projected gradient descent~\cite{venkatesh2023gaussian}. Kleinman \textit{et al.}~\cite{kleinman2021redundant} proposed a measure of redundant information based on the result of a learning problem.
Inspired by PID, McSharry \textit{et al.}~\cite{mcsharry2024learning,kaplanis2023learning} proposed a representation learning architecture to learn highly synergistic features from data.
More generally, there is a wide range of variational methods that turn the estimation of mutual information into an optimisation problem~\cite{poole2019variational}.

\section{Conclusion} \label{sec:conclusions}

We introduced the \W- and \Minfo, novel information-theoretic measures of lower- and higher-order information structure, tailored for use on large dynamical systems. 
In addition to formulating the calculation of \W as a convex optimisation problem, we designed a robust and efficient algorithm that enables practical applications to large, high-dimensional systems. 
We illustrated their behaviour on synthetic systems with Gaussian and non-linear dynamics, showing their robustness to noise and capability to index critical behaviour.
Finally, we employed \Minfo to analyse empirical brain activity data of macaques in different states of consciousness and mice during a visual discrimination task. Our results showed that higher values of \M reflect conscious levels and accurate perceptual performance, demonstrating that our proposed measures can detect meaningful phenomena in real-world scenarios. Moreover, as an additional theoretical contribution, we provided an information decomposition framework based on \W- and \Minfo, enabling a finer-grained analysis of multivariate information dynamics.

Although \Winfo is defined for any probability distribution, here we have focused only on the case where $\bX,\bY$ are jointly Gaussian, proving that it provides a tight bound on the true \W (Appendix). However, to compute \W in other scenarios and potentially make calculations in the Gaussian case faster, future work could design more efficient optimisers that better exploit the geometry of $\Delta_P$ via Riemannian machine learning~\cite{geoopt2020kochurov} or generalise to other distributions~\cite{pakman2021estimating}. Moreover, a limitation of this approach is the range of validity of Condition~\ref{cond:existence}, which warrants further investigation. In fact, when this condition does not hold, then \W and \Minfo become upper and lower bounds to the lower- and higher-order information, respectively. 
An additional potential limitation is the use of the inclusion–exclusion principle (IEP), whose application has been criticised~\cite{kolchinsky2022novel}. In our framework, however, the IEP is only invoked to establish the link between \W and lower-order dependencies, while the definition of \Minfo as a measure of higher-order information remains independent of it. A promising direction for future work is therefore to investigate \M as a standalone estimator of higher-order dependencies, disentangled from lower-order contributions.
Finally, future work can also focus on whether \W and \M can be used as objective functions for learning in deep learning systems~\cite{makkeh2025general}.

Taken together, our results demonstrate that \Minfo is a powerful estimator of higher-order information, grounded in solid mathematical and operational foundations, and whose efficient computational scalability enables broad applicability across diverse real-world domains, including neuroscience and beyond. %
By bridging principled information-theoretic foundations with practical feasibility, \Minfo opens the door to a new generation of tools for quantifying collective interactions in high-dimensional complex systems.

\section*{Acknowledgements}
We thank Maximilian Kathofer for useful discussions. The work of F.R. has been supported by the UK ARIA Safeguarded AI programme, the PIBBSS Affiliateship programme, and Open Philanthropy.

\section*{Code availability}
The Python code for calculating \W-, \Minfo, and BROJA-\phiid is publicly available and can be found at \url{https://github.com/alberto-liardi/wimfo}. 

\bibliographystyle{ieeetr}
\bibliography{refs.bib}

\newpage 

\setcounter{lemma}{0}
\setcounter{assumption}{0}
\setcounter{condition}{0}
\setcounter{definition}{0}

\appendix
\booltrue{inappendix}

\section{Optimisation algorithm for \Winfo} \label{app:optim_procedure}

Here we describe in full detail the optimisation procedure implemented for the calculation of \Winfo. 

Let's consider a joint system $(\bX, \bY)$ that follows a Gaussian distribution with full-rank covariance $\Sigma$.
We denote by $P$ the original distribution of the system, and by $Q$ the distribution being optimised. 
To differentiate between the two, we use superscripts: quantities referring to $Q$ will carry a superscript (e.g.\ $\Sigma^Q$), while those associated with the original distribution $P$ will be written without any (e.g.\ $\Sigma$).

Given the domain $\Delta_P$ that contains all probability distributions $Q$ such that $Q(X_i,Y_j) = P(X_i,Y_j)$, we remind that the objective of the optimisation is to find the distribution $Q^*\in\Delta_P$ that minimises the mutual information $I_Q(\bX,\bY)$. For Gaussian variables, this has the form:
\begin{equation} \label{eq:MI_gaussian_Q}
    \begin{split}
     I_Q(\bX,\bY) & = \dfrac12\log\left(\frac{\det(\Sigma_X^Q)\det(\Sigma_Y^Q)}{\det(\Sigma^Q)} \right) \,.
    \end{split}
\end{equation}
Therefore, we look for a convenient parametrisation of the covariance matrices of the probability distributions in $\Delta_P$.

We start from the full covariance of the original system, which can be written in blocks as:
\begin{align} \label{eq:full_cov_block}
    \Sigma = \begin{pmatrix} \Sigma_{XX} & \Sigma_{XY} \\ \Sigma_{YX} & \Sigma_{YY} \end{pmatrix} ~ ,
\end{align}
where $\Sigma_{XY} = \Sigma_{YX}^{\top}$. 
Without loss of generality, we assume that all $X_i$ and $Y_j$ are zero-mean and unit-variance, such that all diagonal elements of $\Sigma_{XX}$ and $\Sigma_{YY}$ are 1. This is possible because mutual information is invariant under isomorphisms in both arguments~\cite{cover1999elements}.

Instead of working directly on the correlation matrices, we consider their Cholesky decomposition. 
Every lower-triangular matrix $M$ with positive diagonal uniquely determines a symmetric positive definite matrix $S$ through the bijective mapping:
\begin{subequations} \label{eq:SM}
\begin{align}
    M \rightarrow S &: S = M M^\top \label{eq:SMa} \\
    S \rightarrow M &: M = \text{chol}(S) \label{eq:SMb} \,,
\end{align}
\end{subequations}
where $\text{chol}(S)$ is the Cholesky decomposition of $S$.

Hence, let $L = \text{chol}(\Sigma)$ be the Cholesky decomposition of the full system (i.e.\ $\Sigma = L L^\top$). We can write this in blocks:
\begin{align}
    \Sigma = L L^\top = \begin{pmatrix} L_{XX} & 0 \\ L_{YX} & L_{YY} \end{pmatrix} \begin{pmatrix} L_{XX}^\top & L_{YX}\top \\ 0 & L_{YY}^\top \end{pmatrix} = \begin{pmatrix} L_{XX} L_{XX}^\top & L_{XX}L_{YX}^\top \\ L_{YX}L_{XX}^\top & L_{YX}L_{YX}^\top + L_{YY}L_{YY}^\top \end{pmatrix}~ ,
    \label{eq:cholblock}
\end{align}
where $L_{XX}$ and $L_{YY}$ are lower-triangular diagonal-positive matrices, and $L_{XY}$ is any real matrix with entries $\in [-1,1]$. A similar structure holds for $\Sigma^Q$.

To incorporate the constraints of $\Delta_P$, we explicitly enforce that all the cross-covariances between the $X_i$'s and the $Y_j$'s are equal to those in the original system, i.e.\
\begin{align} \label{eq:offdiag_constraint}
    \Sigma^Q_{YX} = \Sigma_{YX} ~ .
\end{align}
Thus, comparing the off-diagonal terms of Eqs.~\eqref{eq:full_cov_block}-\eqref{eq:cholblock} and imposing Eq.\eqref{eq:offdiag_constraint}, we obtain:
\begin{align} \label{eq:chol_constraint}
    \Sigma^Q_{YX} = L^Q_{YX} L^{Q,\top}_{XX} = \Sigma_{YX} \implies
    L^{Q}_{YX} = \Sigma_{YX} \left(L^{Q,\top}_{XX}\right)^{-1} \,.
\end{align}
This suggests that any Cholesky matrix $L^Q$ of a distribution $Q\in\Delta_P$ must have the following block structure:
\begin{align} \label{eq:chol_delta_P_app}
    L^Q = \begin{pmatrix} L^Q_{XX} & 0 \\ \Sigma_{YX} (L^{Q,\top}_{XX})^{-1  } & L^Q_{YY} \end{pmatrix} \,.
\end{align}
There is one more constraint we need to satisfy: the fact that the diagonal of $\Sigma^Q$ must be equal to 1. We achieve this by employing parametrisation techniques from the Parametrization Cookbook (PC)~\cite{leger2023parametrization}. 

We begin with $\bX$. Since the covariance of $\bX$ is simply ${\Sigma^Q_{XX} = L_{XX}^Q L_{XX}^{Q,\top}}$, we can make sure the above constraint is satisfied by parametrising the space of correlation matrices in $\bX$~\cite[Sec.~2.4.5]{leger2023parametrization}.

Regarding $\bY$, its covariance is $\Sigma^Q_{YY} = L^Q_{YX} L^{Q,\top}_{YX} + L^Q_{YY}L^{Q,\top}_{YY}$. Since $L^Q_{XY}$ is determined by $L^Q_{XX}$ and $\Sigma_{YX}$, we want to parametrise $L^Q_{YY}$ such that the diagonal of this sum is 1. Let $\bm d = \diag(L^Q_{YX} L^{Q,\top}_{YX})$ be the vector of diagonal elements of the first summand above.
As per PC, if $M$ is a lower-triangular matrix and $S = M M^\top$, for $S_{ii}$ to be equal to $a$, the $i$-th row of $M$ must be a point in the half-sphere of radius $\sqrt{a}$. By inspecting the expression for $\Sigma^Q_{YY}$ above, we conclude that $L^Q_{YY}$ can be parametrised as a lower-triangular matrix where the $i$-th row is a point in the half-sphere of radius $\sqrt{1 - d_{i}}$.

Hence, this enables us to parametrise the optimisation set $\Delta_P$ by a correlation Cholesky factor, $L^Q_{XX}$, and a Cholesky factor with fixed diagonal, $L^Q_{YY}$. In turn, these can be easily mapped to a vector of real numbers~\cite{leger2023parametrization}.  

In conclusion, starting from a point in a high-dimensional real space, we map the real numbers to $L_{XX}^Q$ and $L_{YY}^Q$, calculate $L^Q$ and thus $\Sigma^Q$, and finally compute $I_Q(\bX; \bY)$. Since all these maps are differentiable, we can perform the optimisation using standard gradient-based approaches. For additional clarity, we provide in Algorithm~\ref{alg:W-minimisation} a synthetic pseudocode for the minimisation algorithm described above.

\begin{algorithm}[H]
\setstretch{1.35}
\caption{Estimation of $\mathcal{W}$: Minimisation of $I_Q(\bm X,\bm Y)$ using Adam}
\label{alg:W-minimisation}
\begin{algorithmic}[1]
\Require Correlation matrix $\Sigma(\bm{X}, \bm{Y}) \in \mathbb{R}^{n\times n}$
\Ensure Estimated value of $\mathcal{W}(\bm{X}; \bm{Y}) \in \mathbb{R}$
\State Compute Cholesky factorisation: $L \gets \text{chol}(\Sigma)$
\State Map $L$ to real parameters: $\bm \theta_0 = (\bm\theta_{0_x},\bm\theta_{0_y})^\top\gets f(L)$ \Comment{\cite[Sec. 2.4.5]{leger2023parametrization}}
\State Initialise parameters: $\bm \theta = (\bm\theta_{x},\bm\theta_{y})^\top\gets \bm \theta_0$
\Repeat
    \State Reconstruct $L^Q_{XX}$ from $\bm\theta_x$: $L^Q_{XX} \gets f^{-1}(\bm\theta_x)$
    \State Compute $L^Q_{YX}$: $L^Q_{YX} \gets \Sigma_{YX} (L^{Q,\top}_{XX})^{-1}$
    \State Reconstruct $\tilde{L}^Q_{YY}$ from $\bm\theta_y$: $\tilde{L}^Q_{YY} \gets f^{-1}(\bm\theta_y)$
    \State Rescale rows: $(L^Q_{YY})_{i,:} \gets (\tilde{L}^Q_{YY})_{i,:} \cdot \sqrt{1 - \mathrm{diag}(L^Q_{YX} L^{Q,\top}_{YX})_i}$
    \State Compose full Cholesky factor: 
    $L^Q \gets 
        \begin{pmatrix} 
            L^Q_{XX} & 0 \\ 
            L^Q_{YX} & L^Q_{YY} 
        \end{pmatrix}
    $
    \State Compute $\Sigma^Q \gets L^Q L^{Q,\top}$
    \State Evaluate objective: $I_Q(\bm X, \bm Y) \gets I_Q(\Sigma^Q)$
    \State Compute gradient: $g_t \gets \nabla_{\bm\theta} I_Q$
    \State Update parameters: $\bm\theta \gets \text{Adam}(\bm\theta, g_t)$
\Until{convergence criterion is met}
\end{algorithmic}
\end{algorithm}

\section{Derivations and properties of \W- and \Minfo} \label{app:WM_properties}

In this section, we prove the formal link between \W- and \Minfo and their formulation presented in Sec.~\eqref{sec:results_math}. %
Then, we discuss some properties of the probability distribution set $\Delta_P$, and show that the mutual information $I(\bX;\bY)$ is a convex function of the marginal covariances. 

In the following paragraphs, we assume $\bX,\bY$ to be Gaussian-distributed random variables with mean zero and correlation matrices $\Sigma_{XX},\Sigma_{YY}$, respectively.

\subsection{Definition of \W- and \Minfo}
In Sec.~\ref{sec:results_math}, we provided the mathematical definitions for \W- and \Minfo, and then claimed they correspond to measures of lower- and higher-order interdependences in Proposition 1 under Assumptions~\ref{ass:low}-\ref{ass:nonneg} and Condition~\ref{cond:existence}. In this paragraph, we present the rigorous proof of the Proposition. 

For completeness, we first recall here Assumptions~\ref{ass:low}-\ref{ass:nonneg}, Condition~\ref{cond:existence}, and Definition~\ref{def:wm}:
\begin{itemize}[leftmargin=*]
    \item[] \begin{assumption}[Pairwise dependence]
        Lower-order information should only depend on the pairwise marginal distributions $P(X_i, Y_j)$.
    \end{assumption}
    \item[] \begin{assumption}[Non-negativity]
        The amount of lower- and higher-order information cannot be negative.
    \end{assumption}
    \item[] \begin{condition}[Existence]
        For a given $P\in\Delta$, there exists a $Q^*$ with  ${Q^*(X_i,Y_j)=P(X_i, Y_j)}$ that has no higher-order information.
    \end{condition}
    \item[]
    \begin{definition} 
    Given two random vectors $\bX$ and $\bY$, the \W- and \Minfo are defined as:
    \begin{align}
        \W(\bX; \bY) &:= \min_{Q\in\Delta_P} I_Q(\bX;\bY) \,, \label{eq:Winfo_def}\\
        \M(\bX; \bY) &:= I(\bX; \bY) - \W(\bX; \bY) \,,\label{eq:Minfo_def}
    \end{align}
    with $\Delta_P \defeq \{Q\in\Delta : Q(X_i, Y_j) = P(X_i, Y_j) ~ \forall i,j\} \,$.
\end{definition}
\end{itemize}

We then have the following result: \\
\propWdef* %
\begin{proof}
    We first establish basic properties of the \W and \Minfo as defined in Eqs.~\eqref{eq:Winfo_def}-\eqref{eq:Minfo_def}. Since the feasible set
    \begin{equation}
        \Delta_P=\{Q\in\Delta:Q(X_i,Y_j)=P(X_i,Y_j)\ \forall i,j\}
    \end{equation}
    is completely determined by the pairwise marginals of $P$, \W depends only on such pairwise marginals and thus satisfies Assumption~\ref{ass:low}. Moreover, ${\mathcal W(\bX;\bY)\ge 0}$ because mutual information is nonnegative, and ${\mathcal M(\bX;\bY)\ge 0}$ as ${\mathcal W(\bX;\bY)\le I(\bX;\bY)}$. Therefore, \W and \M also satisfy Assumption~\ref{ass:nonneg}.

    We now show that, under Condition~\ref{cond:existence}, $\mathcal W$ coincides with the \textit{true} lower-order information ($\mathcal W_{true}$), and $\mathcal M$ with the \textit{true} higher-order information ($\mathcal M_{true}$). 
    By the inclusion-exclusion principle, the mutual information for a generic distribution $Q\in\Delta_P$ can be written as 
    \begin{equation} \label{eq:MI_low_high_decomp_Q}
    \begin{split}
        I_Q(\bX; \bY) & = \mathcal{W}_{true\,Q}(\bX;\bY)+\mathcal{M}_{true\,Q}(\bX;\bY) \\
        & = \mathcal{W}_{true}(\bX;\bY)+\mathcal{M}_{true\,Q}(\bX;\bY)  \,,
    \end{split}
    \end{equation}
    where the second step follows from the fact that lower-order information only depends on the pairwise marginal distributions, and is therefore constant in $\Delta_P$ (Assumption~\ref{ass:low}).
    By minimising the LHS and RHS of Eq.~\eqref{eq:MI_low_high_decomp_Q}, we obtain
    \begin{equation} \label{eq:IQ_W_low_proof}
    \begin{split}
        \min_{Q\in\Delta_P}I_Q(\bX; \bY) & = \min_{Q\in\Delta_P}(\mathcal{W}_{true\,Q}(\bX;\bY)+\mathcal{M}_{true\,Q}(\bX;\bY)) \\
        & = \mathcal{W}_{true}(\bX;\bY) + \min_{Q\in\Delta_P}\mathcal{M}_{true\,Q}(\bX;\bY) \\ 
        & = \mathcal{W}_{true}(\bX;\bY) \,,
    \end{split}
    \end{equation}
    where the last step follows from Assumption~\ref{ass:nonneg} and Condition~\ref{cond:existence}, which ensure that $\exists\, Q^*\in\Delta_P$ s.t. $\min_{Q\in\Delta_P}\mathcal{M}_{true\,Q}(\bX;\bY)=\mathcal{M}_{Q^*}(\bX;\bY)=0$. Hence, comparing Eq.~\eqref{eq:IQ_W_low_proof} and Eq.~\eqref{eq:Winfo_def} proves that \W captures the lower-order information of the system. 
    
    The fact that \Minfo is the higher-order information then follows directly from the inclusion-exclusion principle, for which 
    \begin{equation} \label{eq:MI_low_high_decomp}
        I(\bX; \bY) = \mathcal{W}(\bX;\bY)+\mathcal{M}(\bX;\bY) \,.
    \end{equation}
    
    If, on the other hand, Condition~\ref{cond:existence} is not satisfied, then from Eq.\eqref{eq:IQ_W_low_proof} we have that:
    \begin{equation}
    \begin{split}
        \min_{Q\in\Delta_P}I_Q(\bX; \bY) & = \mathcal{W}_{true}(\bX;\bY) + \min_{Q\in\Delta_P}\mathcal{M}_{true\,Q}(\bX;\bY) \\
        & =: \mathcal{W}(\bX;\bY)\,,
    \end{split}
    \end{equation}
    which therefore shows \W is an upper bound of the lower-order information, i.e.\ $\W\ge\mathcal{W}_{true}$ due to Assumption~\ref{ass:nonneg}. Analogously, for \Minfo we have
    \begin{equation}
        \begin{split}
            \mathcal{M}(\bX;\bY) & := I(\bX;\bY) - \mathcal{W}(\bX;\bY) \\
            & = I(\bX;\bY) - \mathcal{W}_{true}(\bX;\bY) - \min_{Q\in\Delta_P}\mathcal{M}_{true\,Q}(\bX;\bY) \\
            & = \mathcal{M}_{true}(\bX;\bY) - \min_{Q\in\Delta_P}\mathcal{M}_{true\,Q}(\bX;\bY) \\
            & \le \mathcal{M}_{true}(\bX;\bY)\,,
        \end{split}
    \end{equation}
    hence \M is a lower bound of the higher-order information in the system.
\end{proof}

\subsection{Domain $\Delta_P$}
Let's now focus on the optimisation set $\Delta_P$. We have the following result: \\

\begin{proposition}[Convexity of $\Delta_P$] \label{prop:delta_p_convex}
    For $\bX,\bY$ joint Gaussian random variables, the probability distribution domain 
    \begin{equation} \label{eq:delta_p_app}
        \Delta_P \defeq \{Q\in\Delta : Q(X_i, Y_j) = P(X_i, Y_j) ~ \forall i,j\}
    \end{equation}
    is a convex set. 
\end{proposition}
\begin{proof}
The space $\Delta$ of all probability distributions of $\bX,\bY$ corresponds to the space of all correlation matrices, which defines a region in the parameter space called \textit{elliptope} -- 
a convex subset of Euclidean space bounded by elliptical surfaces. 
Since the optimisation domain $\Delta_P$ is obtained by fixing the pairwise marginal distributions within $\Delta$ (Eq.~\eqref{eq:delta_p_app}), this consists in a slicing operation, which maintains the convexity property of the subset $\Delta_P$. 
\end{proof}

Hence, as explained in Appendix~\ref{app:optim_procedure}, we can then parametrise the correlation matrices in $\Delta_P$ by using Eq.~\eqref{eq:chol_delta_P_app}, which imposes the constraints on the marginal distributions for any valid Cholesky factors $L_{XX}$ and $L_{YY}$. 
However, for computational purposes, there is one additional requirement to satisfy. Since we want $\Sigma_{YY}$ to be a correlation matrix, Eq.~\eqref{eq:cholblock} implies $\diag(L_{YX}L_{YX}^\top + L_{YY}L_{YY}^\top)=1$, which entails an extra constraint on the parameters of $L_{XX}$: \\
\begin{lemma}[Optimisation constraint] \label{lemma:add_constraint}
For any $Q\in\Delta_P$, the Cholesky factor $L_{XX}^Q$ of Eq.~\eqref{eq:chol_delta_P_app} needs to satisfy 
\begin{equation}
    \diag(L_{YX}L_{YX}^\top) = \diag(\Sigma_{YX}\Sigma_{XX}^{Q,-1}\Sigma_{YX}^{\top})<1\,, 
\end{equation}
where $\Sigma_{XX}^Q=L_{XX}^QL_{XX}^{Q,\top}$.
\end{lemma}

Although in theory Lemma~\ref{lemma:add_constraint} constrains our optimisation procedure by imposing an additional condition on $L_{XX}^Q$, we see that in practice this is naturally resolved by the properties we prove below.

\subsection{\texorpdfstring{Convexity of $I(\bX; \bY)$}{Convexity of I(X;Y)}}

We now turn to the properties of the mutual information $I(\bX; \bY)$. For ease of notation, in the rest of this paragraph we denote with $\SX$ and $\SY$ the correlation matrices of $\bX,\bY$, with dimensions $d_X\times d_X$ and $d_Y\times d_Y$ respectively, and with $A$ the cross-correlation term $\Sigma_{XY}$. Hence the full covariance matrix of the system reads
\begin{equation}
    \Sigma = \begin{pmatrix}
        \SX & A \\
        A^\top & \SY 
    \end{pmatrix} \,.
\end{equation}

We start by writing explicitly the expression of the mutual information~\cite{cover1991information} for Gaussian variables: 
\begin{equation} \label{eq:MI_gaussian}
    \begin{split}
     I(\bX;\bY) & = \dfrac12\log\left(\frac{\det(\SX)\det(\SY)}{\det(\Sigma)} \right) \\
     & = \dfrac12\log\left(\frac{\det(\SY)}{\det(\SY-A\SX^{-1}A^\top)} \right) \,,
    \end{split}
\end{equation}
where in the second step we made use of the Schur-complement property
\begin{equation}
    \det(\Sigma)=\det(\SX)\det(\SY-A\SX^{-1}A^\top) \,.
\end{equation}
Thus, the objective function of the optimisation is a real-valued function of positive definite matrices
\begin{equation}
\begin{aligned}
    f\colon &\; S_{d_X}^{++} \times S_{d_Y}^{++} \to \mathbb{R} \\
            &\; (\SX, \SY) \mapsto I(\bX; \bY) \,,
\end{aligned}
\end{equation}
where $S_n^{++}$ is the space of positive definite $n\times n$ matrices.

In the following, we prove that $f$ is convex. However, to reach the main Proposition, we first need some results on positive definite matrices. \\

\begin{definition}
A symmetric matrix $A \in \mathbb{R}^{n \times n}$ is \emph{positive definite} if, for all non-zero vectors $x \in \mathbb{R}^n$, the following holds:
\begin{equation}
    x^\top A x > 0.
\end{equation}
\end{definition}

Additionally, an ordering structure can be imposed on the space $\posdef$: \\
\begin{definition}
Let $A, B \in \posdef$. We say that $A$ is strictly smaller than $B$ in the Loewner sense if their difference is positive definite, i.e.\
\begin{equation}
A \prec B \quad\text{if}\quad B - A \in \posdef \,.
\end{equation}
\end{definition}
Using this partial ordering, we can indicate that $A$ is a positive definite matrix by writing $A\succ0$. 
    
With these definitions at hand, we can continue in our derivation. We start by reminding of a standard result: \\
\begin{lemma}[Square of positive definite matrix]\label{lemma:square_posdef}
A matrix $A\in \mathbb{R}^{n\times n}$ is positive definite if and only if it has a unique positive definite square root $A^{1/2}$:
\begin{equation}
    A \in \posdef \iff \exists ! \, A^{1/2}\in\posdef \text{ s.t. } A^{1/2}A^{1/2} = A \,.
\end{equation}
\end{lemma}
\begin{proof}
    See e.g.\ Refs.~\cite{horn2012matrix, bhatia2009positive}.
\end{proof}

Directly following from Lemma~\ref{lemma:square_posdef}, we have a useful property:\\
\begin{lemma} \label{lemma:ABA_posdef}
    Given $A,B\in\posdef$, then $A^{1/2}BA^{1/2}\in \posdef$.
\end{lemma}
\begin{proof}
    Consider the matrix $C = A^{1/2} B A^{1/2}$. We show that $C$ is symmetric and positive definite.
    
    Since $A^{1/2}$ and $B$ are symmetric (Lemma~\ref{lemma:square_posdef} and hypothesis), we have
    \begin{equation}
    C^\top = (A^{1/2} B A^{1/2})^\top = A^{1/2} B^\top A^{1/2} = A^{1/2} B A^{1/2} = C \,,
    \end{equation}
    so $C$ is symmetric.
    
    Then, let $x \in \mathbb{R}^n$, $x \neq 0$, and define $y = A^{1/2} x$. Since $A^{1/2}$ is invertible, $y \neq 0$. Then,
    \begin{equation}
    x^\top C x = x^\top A^{1/2} B A^{1/2} x = y^\top B y > 0 \,,
    \end{equation}
    because $B \in \posdef$.
    Hence, $C = A^{1/2} B A^{1/2} \in \posdef$.
\end{proof}

Although the product of positive definite matrices is not in general positive definite, interesting properties hold for the trace of products:\\
\begin{lemma} \label{lemma:TrAB_ABAC}
    Given $A,B,C\in\posdef$, we have
\begin{align}
    & \text{a)} \quad \Tr(AB)>0 \\
    & \text{b)} \quad \Tr(ABAC)>0 \,.
\end{align}
\end{lemma}
\begin{proof}
a) We have
\begin{equation}
\begin{split}
    \Tr(AB) & \alphaeq \Tr(A^{1/2}A^{1/2}B) \\ 
    & \betaeq \Tr(A^{1/2}BA^{1/2}) \\
    & \gammadiseq 0 \,,
    \end{split}
\end{equation}
where in (a) we used Lemma~\ref{lemma:square_posdef}, in (b) the cyclic property of the trace, and (c) follows from $A^{1/2}BA^{1/2}$ being positive definite (Lemma~\ref{lemma:ABA_posdef}).

b) We can write
\begin{equation}
\begin{split}
    \Tr(ABAC) & = \Tr(A^{1/2}BA^{1/2}A^{1/2}CA^{1/2}) \\
    & = \Tr(\tilde{B}\tilde{C}) \\
    & > 0 \,,
\end{split}
\end{equation}
where we introduced the positive definite matrices $\tilde{B}\equiv A^{1/2}BA^{1/2}$ and $\tilde{C}\equiv A^{1/2}CA^{1/2}$, and used part a) and Lemma~\ref{lemma:ABA_posdef}. 
\end{proof}

Finally, we provide a result on the monotonicity of the trace operator:\\
\begin{lemma} \label{lemma:tr_squares}
    Consider $A,B\in\posdef$ so that $A\succ B$. Then $\Tr(A^2)>\Tr(B^2)$.
\end{lemma}
\begin{proof}
    Let $C\equiv A-B$. By definition of $A\succ B$, we have that $C\in\posdef$. 
    We can then write
    \begin{equation} \label{eq:trace_square}
    \begin{split}
        \Tr(A^2)-\Tr(B^2) & = \Tr((B+C)^2)-\Tr(B^2) \\
        & \alphaeq 2\Tr(BC)+\Tr(C^2) \\
        & \betadiseq 0 \,,
    \end{split}
    \end{equation}
    where in (a) we expanded $A = B+C$ and used linearity of the trace operator, and (b) follows from positive definiteness of $B,C$ (Lemma~\ref{lemma:TrAB_ABAC}).
\end{proof}

We are now ready to state the main result:\\

\begin{proposition} \label{prop:Wconvex}
    Given $\bX, \bY$ multivariate jointly Gaussian random variables with non-singular marginal and cross-covariances, their \Winfo is a convex function in the marginal covariances of $\bX, \bY$. \\
\end{proposition}
\begin{proof}
Since $f$ is convex if and only if it is convex along any direction in its domain, it is sufficient to show convexity along any line. 
Using a real parameter $t$, we parametrise such a generic line as 
\begin{equation} \label{eq:t_parametrisation}
    \SX = X_0 + tX ,\quad \SY = Y_0 + tY \,,
\end{equation}
where $X_0, X, \in \posdef[d_X]$ and $Y_0, Y \in \posdef[d_Y]$, with $d_X\times d_X$ and $d_Y\times d_Y$ the dimensions of $\Sigma_X$ and $\Sigma_Y$, respectively.

Therefore, we define $g(t)\defeq f( X_0 + tX,  Y_0 + tY)\,\, \forall t>0 \,\,\text{s.t. } X_0 + tX\succ 0 \text{ and } Y_0 + tY\succ 0$. 
Using this manipulation, we are now dealing with a real-valued function of a real variable:
\begin{equation}
\begin{aligned}
    g\colon &\; \mathbb{R} \to \mathbb{R} \\
            &\; t \mapsto I(\bX; \bY) \,.
\end{aligned}
\end{equation}
Using Eq.~\eqref{eq:MI_gaussian}, we can write $g$ explicitly as
\begin{equation}
    g(t) = \log\left(\det(\SY)\right) - \log \left(\det(\SY - A \SX^{-1} A^\top)\right) \,, 
\end{equation}
where $\Sigma_X$ and $\Sigma_Y$ should be interpreted as functions of $t$ as per Eq.~\eqref{eq:t_parametrisation}. 

Following~\cite{boyd2004convex, marshall1979inequalities}, we prove strict convexity by showing that the second derivative $d^2g(t)/dt^2$ is positive $\forall t>0$.

\textit{First Derivative:}
\begin{equation}
    \frac{dg(t)}{dt} = \Tr\left(\SY^{-1}Y \right) - \Tr \left( (\SY-A\SX^{-1}A^\top)^{-1}(Y + A \SX^{-1} X \SX^{-1} A^\top) \right) \,,
\end{equation}
where we used the identities~\cite{petersen2008matrix}:
\begin{gather}
    \frac{d(\log(\det(X))}{d\alpha} = \Tr \left(X^{-1} \frac{dX}{d\alpha} \right) \,,\\
    \frac{dX^{-1}}{d\alpha} = - X^{-1} \frac{dX}{d\alpha} X^{-1} \,.
\end{gather}

\textit{Second Derivative:}
\begin{equation} \label{eq:second_der}
\begin{split}
    \frac{d^2g(t)}{dt^2} & = -\Tr \left( \SY^{-1} Y \SY^{-1} Y \right) + \\
    & \quad + \Tr \left((\SY-A\SX^{-1}A^\top)^{-1}(Y + A \SX^{-1} X \SX^{-1} A^\top)(\SY-A\SX^{-1}A^\top)^{-1}(Y + A \SX^{-1} X \SX^{-1} A^\top) \right) +\\
    & \quad + 2\Tr \left((\SY-A\SX^{-1}A^\top)^{-1}(A\SX^{-1}X\SX^{-1}X\SX^{-1}A^\top) \right)  \\
    & \alphaeq -\Tr \left( \SY^{-1} Y \SY^{-1} Y \right) + \Tr \left((\SY-A\SX^{-1}A^\top)^{-1}Y(\SY-A\SX^{-1}A^\top)^{-1}Y) \right) +\\
    & \quad + 2\Tr \left((\SY-A\SX^{-1}A^\top)^{-1}Y(\SY-A\SX^{-1}A^\top)^{-1}(A \SX^{-1} X \SX^{-1} A^\top) \right) +\\
    & \quad + \Tr \left((\SY-A\SX^{-1}A^\top)^{-1}(A \SX^{-1} X \SX^{-1} A^\top)(\SY-A\SX^{-1}A^\top)^{-1}(A \SX^{-1} X \SX^{-1} A^\top) \right) +\\
    & \quad + 2\Tr \left((\SY-A\SX^{-1}A^\top)^{-1}(A\SX^{-1}X\SX^{-1}X\SX^{-1}A^\top) \right) \\
    & \betadiseq 2\Tr \left((\SY-A\SX^{-1}A^\top)^{-1}Y(\SY-A\SX^{-1}A^\top)^{-1}(A \SX^{-1} X \SX^{-1} A^\top) \right) +\\
    & \quad + \Tr \left((\SY-A\SX^{-1}A^\top)^{-1}(A \SX^{-1} X \SX^{-1} A^\top)(\SY-A\SX^{-1}A^\top)^{-1}(A \SX^{-1} X \SX^{-1} A^\top) \right) +\\
    & \quad + 2\Tr \left((\SY-A\SX^{-1}A^\top)^{-1}(A\SX^{-1}X\SX^{-1}X\SX^{-1}A^\top) \right) \\
    & \gammadiseq 0 \,.
    \end{split}
\end{equation}
In (a) we expanded the second expression and used the cyclic property of the trace to sum two equal terms, whereas in (b) we noticed that
\begin{equation} \label{eq:trace_ineq}
    \Tr \left((\SY-A\SX^{-1}A^\top)^{-1}Y(\SY-A\SX^{-1}A^\top)^{-1}Y) \right) > \Tr \left( \SY^{-1} Y \SY^{-1} Y \right) \,.
\end{equation}
To see this, let's introduce the shorthand $S\equiv(\SY-A\SX^{-1}A^\top)$, i.e.\ $S$ is the Schur-complement of $\Sigma$ w.r.t.\ to $\Sigma_X$. Since $\Sigma$ is positive definite and $\SX$ is invertible, it follows that $S$ is also positive definite~\cite{zhang2006schur}. 
We can therefore rewrite the traces of Eq.~\eqref{eq:trace_ineq} as 
\begin{align}
        \Tr(S^{-1}YQ^{-1}Y)  &= \Tr(Y^{1/2}S^{-1}Y^{1/2}Y^{1/2}S^{-1}Y^{1/2}) \equiv \Tr (M^2) \label{eq:trace_M^2} \\
        \Tr(\SY^{-1} Y\SY^{-1} Y)  &= \Tr(Y^{1/2}\SY^{-1} Y^{1/2}Y^{1/2}\SY^{-1} Y^{1/2}) \equiv \Tr (N^2) \label{eq:trace_N^2} \,,
\end{align}
where $M\equiv Y^{1/2}S^{-1}Y^{1/2}$ and $N\equiv Y^{1/2}\SY^{-1}Y^{1/2}$ are now positive definite matrices (Lemma~\ref{lemma:ABA_posdef}).
Moreover, we have that $M\succ N$:
\begin{align} \label{eq:trace_M>N}
    A\SX^{-1}A^\top \succ  0 & \implies \SY-A\SX^{-1}A^\top \prec \SY \\
    & \implies (\SY-A\SX^{-1}A^\top)^{-1}\succ \SY^{-1} \notag \\
    & \implies Y^{1/2}(\SY-A\SX^{-1}A^\top)^{-1}Y^{1/2}\succ Y^{1/2}\SY^{-1}Y^{1/2} \notag \\
    & \,\,\,\,\equiv\,\,\,\, M\succ N \notag  \,,
\end{align}
where the third implication holds as $Y^{1/2}$ is positive definite. 
Therefore, having Eqs.~\eqref{eq:trace_M^2}-\eqref{eq:trace_N^2}-\eqref{eq:trace_M>N}, we can apply Lemma~\ref{lemma:tr_squares} and obtain Eq.~\eqref{eq:trace_ineq}. 

Finally, in the last step (c) we made use of Lemma~\ref{lemma:TrAB_ABAC}.

Thus, we proved that the mutual information is convex along a generic line in the domain $\Delta_P$, and therefore that it is globally convex. 
\end{proof}

The proof presented here assumes that $\SX, \SY, A$ are well-defined full-rank matrices, %
i.e.\ $\SX \in \RealX{d_X}$, $\SY \in \RealX{d_Y}$, and $A \in \mathbb{R}^{d_Y\times d_X}$ and $\text{rank}(\SX)=d_X$, $\text{rank}(\SY)=d_Y$, $\text{rank}(A)=\min (d_X, d_Y)$.
In case these conditions are not met, then the mutual information is not convex. As a counterexample, consider $A_{ij}=0\,\,\,\forall i,j$, then $I(\bX;\bY)=0\,\,\,\forall\,\SX,\SY\in\posdef$, with $n=d_X=d_Y$.

\subsection{Boundaries of $\Delta_P$}

Although the minimisation problem is well-posed, we need to make sure that the optimisation procedure does not step outside the boundaries of $\Delta_P$, where the function is not defined (Lemma~\ref{lemma:add_constraint}). 
Luckily, it is easy to prove that on the boundary of the domain $\Delta_P$, the mutual information becomes infinite, which is our second main result: \\

\begin{proposition} \label{prop:inner}
    Given $\bX, \bY$ multivariate jointly Gaussian random variables with non-singular marginal covariances, the probability distribution $Q^*$ that minimises Eq.~\eqref{eq:Winfo_def} is an inner point of $\Delta_P$. \\
\end{proposition}
\begin{proof}
    In the parametrisation given by Eq.\eqref{eq:cholblock}, the border of $\Delta_P$ is characterised by $\diag(L_{YX}L_{YX}^\top)_i=1$ for some $i=1,...,n$. If one such diagonal entry is 1, it follows that $\diag(L_{YY}L_{YY}^\top)_i=0$ and $\diag(L_{YY})_i=0$, which implies $\det(\Sigma)=\det(L)^2=0$. 
    On the other hand, $\det(\Sigma_{XX})>0$ and $\det(\Sigma_{YY})>0$ by hypotheses, hence on the boundary we obtain
    \begin{equation} \label{eq:MI_diverging}
        I(\bX;\bY) = \dfrac12\log\left(\frac{\det(\SX)\det(\SY)}{\det(\Sigma)} \right) \rightarrow + \infty \,.
    \end{equation} 
    Thus, the distribution $Q^*$ has to be an inner point of $\Delta_P$
    
\end{proof}
Note that with the parametrisation chosen, the hypotheses of Proposition ~\ref{prop:inner} can be further relaxed, as we only need $L_{XX}$ and $L_{YY}$ to be non-singular upper diagonal matrices. In fact, ${\det (L_{XX})>0 \implies \det(\Sigma_{XX})>0}$, and 
\begin{equation}
\begin{split}
    \det(\Sigma_{YY})& =\det(L_{YX}L_{YX}^\top+L_{YY}L_{YY}^\top) \\
    & \alphadisequal \det(L_{YX}L_{YX}^\top)+\det(L_{YY}L_{YY}^\top) \\
    & \betaeq \det(L_{YX}L_{YX}^\top) \\
    & > 0 \,,
\end{split}
\end{equation}
where in (a) we used the matrix concavity of the map $A\mapsto \det(A)$, and in (b) the fact that $L_{YX}L_{YX}^\top$ and $L_{YY}L_{YY}^\top$ are positive definite. The result then follows in the same way from Eq.~\eqref{eq:MI_diverging}.

This result also explains why we can avoid addressing the issue raised in Lemma~\ref{lemma:add_constraint}.
Since $I(\bX;\bY)$ is continuous at any inner point, its gradients must diverge approaching the boundaries. This ensures that standard optimisation techniques such as Stochastic Gradient Descent do not step outside the domain even without directly imposing the extra constraints.

Again, we remind that if $A$ is not full-rank, then the minimum of the mutual information may lie on the boundary as $\det(L_{YX})$ could also vanish. 

Thus, we can present our final result: \\ 
\theoConvex*
\begin{proof}
    Direct from Propositions \ref{prop:Wconvex}-\ref{prop:inner}, and the fact that $\Delta_P$ is a convex set (Proposition~\ref{prop:delta_p_convex}).
\end{proof}

Importantly, we note that the construction outlined in this section implicitly assumes that the minimising distribution $Q^*$ is a joint Gaussian distribution in $\bX$ and $\bY$. We validate this assumption in Appendix~\ref{app:discr_v_gauss}.

\section{Validation of Gaussian minimisation} \label{app:discr_v_gauss}

In this section, we assess the validity of the assumption of joint Gaussianity of the solution which underlies our implementation (Appendix~\ref{app:optim_procedure}). Although this assumption formally yields an upper bound on the true solution, here we show that the bound is tight for toy systems where the exact solution is known. For more general systems, the Gaussian solution closely matches that obtained without imposing joint Gaussianity. 

We first consider the toy models COPY, transfer, and XOR (Sec.~\ref{sec:results_synthetic}), as well as the sixteen systems shown in Fig.~\figsubref{fig:dominant_tests_phiID}{b} and reported in Appendix~\ref{app:broja_phiid}. These systems were explicitly constructed to maximise a specific informational contribution, so their true (expected) solution is known. In all cases, the solution obtained under the joint Gaussian assumption correctly reproduces the true distribution.

We then turn to more general systems and compare the solution obtained under the joint Gaussianity assumption with that obtained without it.
To achieve this, we perform the minimisation of Eq.~\eqref{eq:Winfo_def} in discrete space, where we can relax the Gaussianity constraint. 
Specifically, we generate random zero-mean Gaussian distributions by sampling their covariance matrix from a Wishart distribution. Each distribution is then discretised using bins whose widths are weighted by distribution quantiles. This ensures finer resolution in high-density regions and coarser resolution in low-density areas.
Finally, we perform the minimisation of Eq.~\eqref{eq:Winfo_def} both in Gaussian and discrete space and compare the resulting values ($\mathcal{W}_{gaussian}$ and $\mathcal{W}_{discrete}$). While the Gaussian optimisation is described above (Appendix~\ref{app:optim_procedure}), the discrete minimisation is based on the Iterative Proportional Fitting Procedure (IPFP), which allows us to employ gradient-descent methods (Adam) while rescaling the marginal distributions to satisfy the constraints of $\Delta_P$. The full implementation of the discrete minimisation is available at \url{https://github.com/alberto-liardi/wimfo}.

A direct comparison of the minimising distributions would involve an additional fit of a Gaussian model to the discrete solution, or a discretisation of the Gaussian solution. Since both approaches would introduce additional bias, we use the obtained values of \Winfo as proxies for the comparison, which, after all, are the estimated quantity of interest.
Moreover, after sampling and discretising the starting Gaussian distribution, we exclude cases in which the mutual information differs by more than 20\% before and after the discretisation, which would indicate a poor discrete approximation and therefore strongly affect the estimation of \W. To analyse how our results depend on this filtering, we repeat the procedure with a more rigid 5\% tolerance. 
As expected, the discretisation is highly sensitive to the number of bins used, as evidenced by the stark effect of the cutoffs (Tab.~\ref{tab:discr_gauss}).
Hence, this confirms that mutual information values can change substantially and represent the primary source of discrepancies between $\mathcal{W}_{gaussian}$ and $\mathcal{W}_{discrete}$, and that the above filterings are needed. %

\begin{table}[ht]
\centering
\setlength{\tabcolsep}{20pt} %
\renewcommand{\arraystretch}{1.1} %
\begin{tabular}{c c c}
\toprule
$N$ & 20\% cut & 5\% cut \\
\midrule
5  & 0.00\%  & 0.00\%  \\
10 & 29.18\% & 0.00\%  \\
20 & 70.73\% & 5.27\%  \\
50 & 92.62\% & 58.25\% \\
75 & 93.91\% & 74.35\% \\
\bottomrule
\end{tabular}
\caption{Percentage of randomly generated Gaussian systems whose mutual information remains within 20\% or 5\% of the original value after discretisation. Each value of $N$ indicates the number of bins per dimension.}
\label{tab:discr_gauss}
\end{table}

Comparing $\mathcal{W}_{gaussian}$ and $\mathcal{W}_{discrete}$, we observe that variations in \W between Gaussian and discrete solutions remain limited and always fall below a relative difference of 5\% when the filtering of 20\% is applied (Fig.~\figsubref{fig:discr_v_gauss}{b}). Moreover, such differences significantly decrease when the 5\% cut is performed (Fig.~\figsubref{fig:discr_v_gauss}{c}). Therefore, we argue that the remaining discrepancies are likely attributable to imperfect representation of the Gaussian distribution in discrete space, rather than a limitation imposed by the joint Gaussianity assumption. In fact, given the observed convergence, further increasing the number of bins would be expected to reduce these differences, as it would allow employing a stricter filtering cutoff. Moreover, we note that these results are not due to having considered a limited range of mutual information values, as these span a broad interval (Fig.~\figsubref{fig:discr_v_gauss}{a}).

Finally, even if the residual discrepancies in \W were entirely due to genuine differences between the minimising distributions, we remark that their magnitude would be comparable to the variability induced by finite-sample effects (Fig.~\figsubref{fig:acc_time_tests}{a}). Hence, in practice, a minimisation without the joint Gaussianity assumption would not yield a meaningful improvement in the estimation of \W, and the results obtained by the Gaussian procedure can be considered as the true minima.

\begin{figure}
    \centering
    \includegraphics[width=1\linewidth]{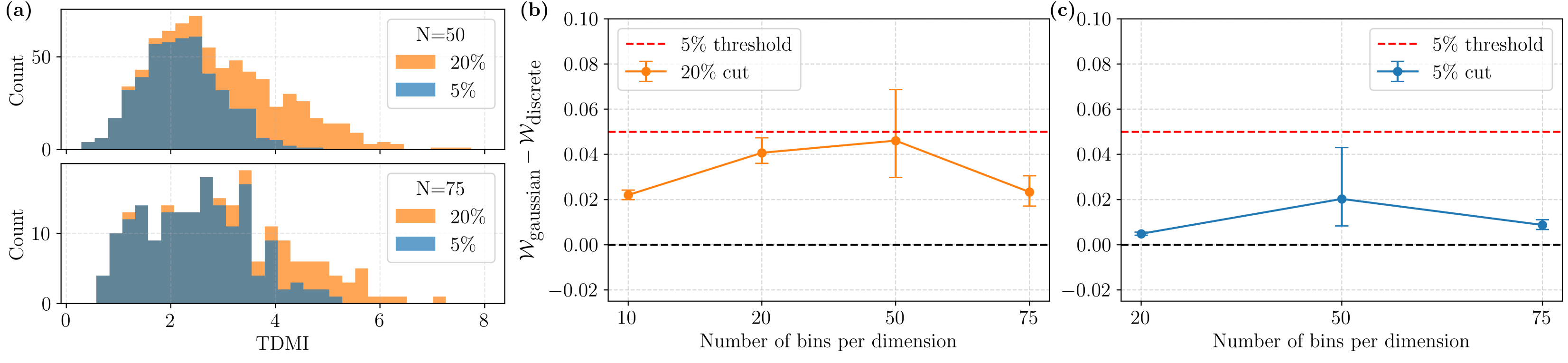}
    \caption{Validation of joint Gaussian solution with discrete systems. (a) Distributions of mutual information values passing the 20\% and 5\% filtering for 50 and 75 bins per dimension. (b) Relative difference between $\mathcal{W}_{gaussian}$ and $\mathcal{W}_{discrete}$ normalised by their respective mutual information for the 20\% filtering. (c) Same as (b) but for the 5\% filtering.}
    \label{fig:discr_v_gauss}
\end{figure}

\section{BROJA-\phiid} \label{app:broja_phiid}

Here we outline the complete derivation of the BROJA-\phiid method presented in Sec.~\ref{sec:broja-phiid}. 
We adopt a practical approach to the framework, focusing on the equations required to compute the \phiid terms. 
For rigorous derivations, theoretical details, and a more general treatment, we refer to the original work \cite{mediano2025toward}.

Given two bivariate random variables $\bX=(X_1,X_2)$ and $\bY=(Y_1,Y_2)$, we aim to decompose the joint mutual information of the system $I(\bX;\bY)$ into redundant ($R$), unique ($U$), and synergistic components ($S$), both in $\bX$ and $\bY$. 
Following~\cite{mediano2025toward}, we start by writing down the functional form of marginal and joint mutual information in terms of the \phiid atoms:
\begin{align}
\allowdisplaybreaks
\label{eq:phiid1}
I(X_1; Y_1) &= \rtr + \rtx + \xtr + \xtx \\[1ex] %
I(X_2; Y_2) &= \rtr + \rty + \ytr + \yty \\[1ex] %
I(X_1; Y_2) &= \rtr + \rty + \xtr + \xty \\[1ex] 
I(X_2; Y_1) &= \rtr + \rtx + \ytr + \ytx \\[1ex] 
I(X_1, X_2; Y_1) &= 
  \begin{aligned}[t]
    &\rtr + \rtx + \xtr + \xtx + \\
    & \ytr + \ytx + \str + \stx
  \end{aligned} \\[1ex] 
I(X_1, X_2; Y_2) &= 
  \begin{aligned}[t]
    &\rtr + \rty + \xtr + \xty+ \\
    & \ytr + \yty + \str + \sty
  \end{aligned} \\[1ex]
I(X_1; Y_1, Y_2) &= 
  \begin{aligned}[t]
    &\rtr + \xtr + \rtx + \xtx +\\
    & \rty + \xty + \rts + \xts
  \end{aligned} \\[1ex]   %
I(X_2; Y_1, Y_2) &= 
  \begin{aligned}[t]
    &\rtr + \ytr + \rtx + \ytx +\\
    &\rty + \yty + \rts + \yts
  \end{aligned} \\[1ex]
I(X_1, X_2; Y_1, Y_2) &=
\begin{aligned}[t]
    &\rtr + \xtr + \ytr + \str + \\
    & \rtx + \xtx + \ytx + \stx + \\
    & \rty + \xty + \yty + \sty + \\
    & \rts + \xts + \yts + \sts \,.
\end{aligned}
\end{align}

Then, six additional constraints are given by six BROJA-PID decompositions on different parts of the system. 
Specifically, if, for instance, we consider $X_1, X_2$ as sources and $Y_1$ as target, we can perform a PID that provides the redundancy $R(X_1, X_2; Y_1)$, which is equal to
\begin{equation}
    R(X_1, X_2; Y_1) = \rtr + \rtx \,.
\end{equation}
Performing a similar procedure on all possible combinations of sources and targets, we obtain the remaining 5 redundancies:
\begin{align}
R(X_1, X_2; Y_1, Y_2) &= \rtr + \rtx + \rty + \rts \\[1ex]
R(X_1, X_2; Y_2) &= \rtr + \rty \\[1ex]
R(Y_1, Y_2; X_1, X_2) &= \rtr + \xtr + \ytr + \str \\[1ex]
R(Y_1, Y_2; X_1) &= \rtr + \xtr \\[1ex]
R(Y_1, Y_2; X_2) &= \rtr + \ytr \,.
\end{align}

Finally, to uniquely determine all 16 atoms, we need an additional condition. We attain this by including all the lower-order terms, i.e.\ those that depend only on redundancy and unique information, in the \Winfo.  
We justify this step as follows:
\begin{enumerate}
    \item redundancies and unique information depend only on the pairwise marginals;
    \item therefore, all atoms that contain only redundancies and unique information are constant in $\Delta_P$.
    \item There is at least one distribution $Q^* \in \Delta_P$ which has only redundancies and unique information, with all other atoms being zero;
    \item this is the distribution that the optimiser will find.
    \item Hence, $\min_{Q\in\Delta_P} I_Q(\bX; \bY)$ equals the sum of all redundancies and unique information.
\end{enumerate}

Thus, we can write
\begin{equation} \label{eq:phiid16}
\begin{split}
    \W =& \rtr + \rtx + \rty +\\ &\xtr + \xtx + \xty +\\ &\ytr + \ytx + \yty \,.
    \end{split}
\end{equation}

Following this reasoning, we can also express \Minfo in terms of the remaining \phiid atoms, obtaining:
\begin{equation}
\begin{split}
    \M =& \rts + \xts + \yts +\\ &\str + \stx + \sty +\\ &\sts \,.
\end{split}
\end{equation}
In other words, \Minfo contains all the atoms that involve synergistic terms, either in $\bX$ or $\bY$.

Thus, the \phiid atoms can be computed by solving the linear system of equations given by Eqs.~\eqref{eq:phiid1}-\eqref{eq:phiid16}.

\section{Methods and data} \label{app:methods_data}

\subsection{Toy models} \label{app:toy_models}
Here we give the details of the toy models presented in Secs.~\ref{sec:results_synthetic}-\ref{sec:broja-phiid}.

Considering two univariate time series $X_t$ and $Y_t$ that follow Gaussian dynamics, we focus on the \phiid decomposition of the mutual information $I(X_t,Y_t;X_{t+1},Y_{t+1})$. 
To validate our methods in controlled scenarios, we construct 16 systems, each designed to maximally express one of the atoms in the decomposition.
Defining an auxiliary Gaussian variable $Z_t\sim\mathcal{N}(0,1)$ and a Gaussian noise term $\varepsilon_t\sim\mathcal{N}(0,10^{-4})$, we repeatedly sample these terms to obtain a time series, and then generate each of the 16 systems as follows, with the dominant atom indicated alongside each configuration:
\begin{enumerate}
    \item \rtr: we correlate both the sources and the targets together:
    \begin{subequations}
    \begin{align}
        X_t & = Z_t + \varepsilon \\
        Y_t & = Z_t + \varepsilon \\
        X_{t+1} & = Z_t + \varepsilon \\
        Y_{t+1} & = Z_t + \varepsilon  \,;
    \end{align}
    \end{subequations}
    \item \rtx: we correlate the sources with the first target:
    \begin{subequations}
    \begin{align}
        X_t & = Z_t + \varepsilon \\
        Y_t & = Z_t + \varepsilon \\
        X_{t+1} & = Z_t + \varepsilon \\
        Y_{t+1} & = \varepsilon  \,;
    \end{align}
    \end{subequations}
    \item \rty: we correlate the sources with the second target:
    \begin{subequations}
    \begin{align}
        X_t & = Z_t + \varepsilon \\
        Y_t & = Z_t + \varepsilon \\
        X_{t+1} & = \varepsilon \\
        Y_{t+1} & = Z_t + \varepsilon  \,;
    \end{align}
    \end{subequations}
    \item \rts: we correlate the targets and make both sources functions of the sum of the targets:
    \begin{subequations}
    \begin{align}
        X_{t+1} & = 100 \varepsilon \\
        Y_{t+1} & = 100 \varepsilon \\
        X_t & = X_{t+1} + Y_{t+1} + \varepsilon \\
        Y_t & = X_{t+1} + Y_{t+1} + \varepsilon \,;
    \end{align}
    \end{subequations}
    \item \xtr: we correlate the first source with the targets:
    \begin{subequations}
    \begin{align}
        X_t & = Z_t + \varepsilon \\
        Y_t & = \varepsilon \\
        X_{t+1} & = Z_t + \varepsilon \\
        Y_{t+1} & = Z_t + \varepsilon  \,;
    \end{align}
    \end{subequations}
    \item \xtx: we correlate the first source with the first target:
    \begin{subequations}
    \begin{align}
        X_t & = Z_t + \varepsilon \\
        Y_t & = \varepsilon \\
        X_{t+1} & = Z_t + \varepsilon \\
        Y_{t+1} & = \varepsilon  \,;
    \end{align}
    \end{subequations}
    \item \xty: we correlate the first source with the second target:
    \begin{subequations}
    \begin{align}
        X_t & = Z_t + \varepsilon \\
        Y_t & = \varepsilon \\
        X_{t+1} & = \varepsilon \\
        Y_{t+1} & = Z_t + \varepsilon  \,;
    \end{align}
    \end{subequations}
    \item \xts: we correlate the targets and make the first source function of the sum of the targets:
    \begin{subequations}
    \begin{align}
        X_{t+1} & = Z_t + \varepsilon \\
        Y_{t+1} & = Z_t + \varepsilon \\
        X_t & = X_{t+1} + Y_{t+1} + \varepsilon \\
        Y_t & = \varepsilon \,;
    \end{align}
    \end{subequations}
    \item \ytr: we correlate the second source with both targets:
    \begin{subequations}
    \begin{align}
        X_t & = \varepsilon \\
        Y_t & = Z_t + \varepsilon \\
        X_{t+1} & = Z_t + \varepsilon \\
        Y_{t+1} & = Z_t + \varepsilon  \,;
    \end{align}
    \end{subequations}
    \item \ytx: we correlate the second source with the first target:
    \begin{subequations}
    \begin{align}
        X_t & = \varepsilon \\
        Y_t & = Z_t + \varepsilon \\
        X_{t+1} & = Z_t + \varepsilon \\
        Y_{t+1} & = \varepsilon  \,;
    \end{align}
    \end{subequations}
    \item \yty: we correlate the second source with the second target:
    \begin{subequations}
    \begin{align}
        X_t & = \varepsilon \\
        Y_t & = Z_t + \varepsilon \\
        X_{t+1} & = \varepsilon \\
        Y_{t+1} & = Z_t + \varepsilon  \,;
    \end{align}
    \end{subequations}
    \item \yts: we correlate the targets and make the second source function of the sum of the targets:
    \begin{subequations}
    \begin{align}
        X_{t+1} & = 100 \varepsilon \\
        Y_{t+1} & = 100 \varepsilon \\
        X_t & = \varepsilon \\
        Y_t & = X_{t+1} + Y_{t+1} + \varepsilon \,;
    \end{align}
    \end{subequations}
    \item \str: we correlate the sources and make both targets functions of the sum of the sources:
    \begin{subequations}
    \begin{align}
        X_t & = 100\varepsilon \\
        Y_t & = 100 \varepsilon \\
        X_{t+1} & = X_t + Y_t + \varepsilon \\
        Y_{t+1} & = X_t + Y_t  + \varepsilon \,;
    \end{align}
    \end{subequations}
    \item \stx: we correlate the sources and make the first target function of the sum of the sources:
    \begin{subequations}
    \begin{align}
        X_t & = 100 \varepsilon \\
        Y_t & = 100 \varepsilon \\
        X_{t+1} & = X_t  + Y_t + \varepsilon \\
        Y_{t+1} & = \varepsilon  \,;
    \end{align}
    \end{subequations}
    \item \sty: we correlate the sources and make the second target function of the sum of the sources:
    \begin{subequations}
    \begin{align}
        X_t & = 100 \varepsilon \\
        Y_t & = 100 \varepsilon \\
        X_{t+1} & = \varepsilon \\
        Y_{t+1} & = X_t  + Y_t +\varepsilon  \,;
    \end{align}
    \end{subequations}
    \item \sts: we correlate the sources and one target and make the other target function of the both sources and target:
    \begin{subequations}
    \begin{align}
        X_t & = 100\varepsilon \\
        Y_t & = 100\varepsilon \\
        X_{t+1} & = 100\varepsilon \\
        Y_{t+1} & = X_t + Y_t - X_{t+1} \varepsilon  \,.
    \end{align}
    \end{subequations}
\end{enumerate}
These systems provide the expected predominant atoms also with other definitions of \phiid (e.g.\ the Minimal Mutual Information formulation introduced in Ref.~\cite{barrett2015exploration} for PID, and generalised to \phiid in Ref.~\cite{mediano2025toward}). Hence, we use them to test both the behaviour of $\M(X_t,Y_t;X_{t+1},Y_{t+1})$, $\W(X_t,Y_t;X_{t+1},Y_{t+1})$ (Sec.~\ref{sec:results_synthetic}), and BROJA-\phiid (Sec.~\ref{sec:broja-phiid}).

\subsection{VAR models} \label{app:VAR_spec}
Vector Autoregression (VAR) is a statistical framework often used in statistics and complexity science to model the evolution of a dynamical system. VAR models have been applied in a variety of areas, such as sociology~\cite{box2014time}, economics~\cite{sims1980macroeconomics}, statistics~\cite{hatemi2004multivariate}, and beyond. 

The advantage of employing VAR models lies in their capability to capture interdependencies in time series data~\cite{barrett2010multivariate}, providing an efficient yet robust methodology to infer the system's information dynamics. 

Specifically, given a multivariate time series $\bX_t=(X_{t}^1, ..., X_{t}^n)^\top$, the defining VAR equation can be written as 
\begin{equation}
    \bX_t = \sum_{l=1}^p A_l \bX_{t-l}+\boldsymbol{\varepsilon}_t \,,
\end{equation}
where $A_l$ is the $n\times n$ matrix of coefficients for the timepoint $t-l$, $\boldsymbol{\varepsilon}$ a Gaussian noise term sampled from $\mathcal{N}(0,\bV)$, and $p$ the VAR model order. Since this system captures contributions up to $p$ timesteps in the past, the model is named VAR(p).

Once the parameters $A$ and $\bV$ are fitted to the data, e.g.\ with standard techniques such as Ordinary Least Squares~\cite{box2015time}, one can use the Yule-Walker equations to estimate the autocovariance matrices of the system~\cite{yule1927vii, walker1931periodicity}, and thus construct the full covariance matrix~\cite{liardi2024null}. 

Since the VAR models considered here are Gaussian stationary processes, we can then calculate \W- and \Minfo as presented in Sec.~\ref{sec:results_implem}. 
For simplicity, all the autoregressive systems considered in this paper are VAR(1) models.

\subsection{Wilson-Cowan model} \label{app:WC_spec}

The Wilson–Cowan (WC) model is a foundational mathematical framework for describing the dynamics of interacting populations of excitatory and inhibitory neurons~\cite{wilson1972excitatory, wilson1973mathematical}. Due to its generality and ease of applicability, the model has become one of the cornerstones of theoretical neuroscience and computational biology~\cite{sejnowski1976global, amit1997model, brunel2000dynamics, renart2007mean, breakspear2017dynamic}, setting the ground for various generalisations~\cite{destexhe2009wilson, coombes2005waves}.

Rather than tracking individual spikes, WC describes the average activity of large populations of excitatory and inhibitory neurons. This coarse-grained approach significantly reduces computational load while retaining enough dynamical richness to reproduce key neural phenomena such as oscillations, bistability, and critical transitions~\cite{ermentrout2010mathematical, cowan2016wilson}.

The basic Wilson–Cowan equations are a pair of coupled non-linear ordinary differential equations (ODEs) that describe how the mean firing rates $E(t)$ and $I(t)$ evolve over time:
\begin{align}
\tau_E \frac{dE(t)}{dt} &= -E(t) + S_E\left( W_{EE} E(t) - W_{EI} I(t) + P_E(t) + \eta_E(t) \right), \\
\tau_I \frac{dI(t)}{dt} &= -I(t) + S_I\left( W_{IE} E(t) - W_{II} I(t) + P_I(t) + \eta_I(t) \right) \,,
\end{align}
where $E(t)$ and $I(t)$ are the mean firing rates of excitatory and inhibitory populations, $\tau_E$ and $\tau_I$ are their respective time constants, $W_{XY}$ denotes the connection strength from population $Y$ to $X$, $P_E(t)$ and $P_I(t)$ represent external inputs, $(\eta_E,\,\eta_I)^\top\sim\mathcal{N}(0,V)$ are Gaussian noise terms, and $S_E(\cdot)$ and $S_I(\cdot)$ are sigmoidal activation functions of the form:
\begin{equation}
S_X(x) = \frac{1}{1 + e^{-\alpha (x - \theta_X)}} \,,
\end{equation}
where $\alpha$ controls the gain of the function, and $\theta_X$ the activation threshold for the species $X$.

For our simulations, the following standard parameters were employed:
\begin{itemize}
    \item $\tau_E=\tau_I=1.0$
    \item $P_E=P_I=0.125$
    \item $W_{EE}=10.0,\,W_{II}=3.0$
    \item $\theta_E=0.2, \,\theta_i=4.0$
    \item $\alpha=1.0$ \,.
\end{itemize}
The coupling strengths $W_{EI}, \,W_{IE}$ were varied across simulations, along with the noise correlation parameter in V:
\begin{equation}
    V = \begin{pmatrix}
        1 & c \\
        c & 1
    \end{pmatrix} \,.
\end{equation}

After simulating the Wilson–Cowan model across a range of parameters, we fitted a Gaussian copula to the joint distribution of excitatory and inhibitory firing rates to estimate the system's full covariance structure. While the underlying dynamics are non-linear and thus deviate from the Gaussianity assumption, prior work in information theory has demonstrated that Gaussian copulas can effectively capture statistical dependencies in such systems~\cite{ince2017statistical}. This is enabled by the invariance of mutual information under smooth invertible transformations, which allows copula-based approaches to retain sensitivity to non-linear interactions.
Finally, \W- and \Minfo were calculated on the system across time by following the algorithm of Sec.~\ref{sec:results_implem}.

\subsection{Neuropixel LFP data} \label{app:data_spec}
We used local field potential (LFP) recordings from the publicly available dataset published by Steinmetz \textit{et al.}~\cite{steinmetz2019distributed}, which includes high-density neural recordings across the mouse brain during a visually guided decision-making task. The dataset was collected using Neuropixels probes inserted into the left hemisphere of head-fixed mice performing a two-alternative unforced choice (2AUC) task. Recordings were made across 39 behavioural sessions from 10 mice, comprising a total of 92 probe insertions.

Each session includes simultaneous recordings from hundreds of neurons and LFP signals across multiple cortical and subcortical regions. A typical probe insertion spanned several anatomically and functionally distinct areas, enabling large-scale, distributed recordings. Task trials involved visual stimuli presented unilaterally, bilaterally, or not at all, with reward contingencies linked to the contrast and location of the stimulus. Mice responded by turning a wheel to indicate the side with higher contrast, or by withholding movement if no stimulus was present. Trials where equal-contrast bilateral stimuli were presented were rewarded randomly. This design allowed dissociation of neural signals related to sensory processing, motor action, and choice.

In addition to active task trials, the dataset includes passive viewing trials, conducted in the absence of task demands or reward. During these sessions, the same visual stimuli were presented while the animal remained disengaged, allowing comparison of evoked neural responses across different behavioural contexts.

Further details on surgical preparation, probe configuration, spike sorting, and behavioural training are described in the original study~\cite{steinmetz2019distributed}.

\section{Additional results}
In this paragraph, we report further results for the synthetic systems, the bias and scalability analysis, and the monkey and mouse brain data. 

\subsection{Bias and scalability study details}
The analyses for the bias corrections and scalability were run on the High Performance Computing (HPC) cluster of Imperial College London, United Kingdom. 
In the case of the bias correction, we first sample 1000 covariance matrices from a Wishart distribution, we generate time series of different lengths by sampling from the original covariance, estimate the covariance from the data, and finally compute \Winfo. We then repeat the process 100 times with different covariances and compute mean and Standard Error of the Mean (SEM) for each sample size (Fig.~\ref{fig:acc_time_tests}). 
For the scalability study, 1000 covariance matrices were randomly sampled from a Wishart distribution for each dimension, and then mean and SEM were estimated. 
In both cases, the optimisations were run as parallel jobs on the HPC, each requesting a maximum of 2 CPUs and 4 GB of memory.

To further analyse the effect of bias correction on the estimation of \Winfo, we compare the bias-corrected and non-bias-corrected results, noting that the correction significantly improves the convergence speed (Fig.~\ref{fig:no_bias_corr}).

\begin{figure}[t]
\centering
\begin{subfigure}{0.49\textwidth}
    \centering
    \hspace*{5mm}\figuretitle{No bias correction}
    \includegraphics[width=1\linewidth]{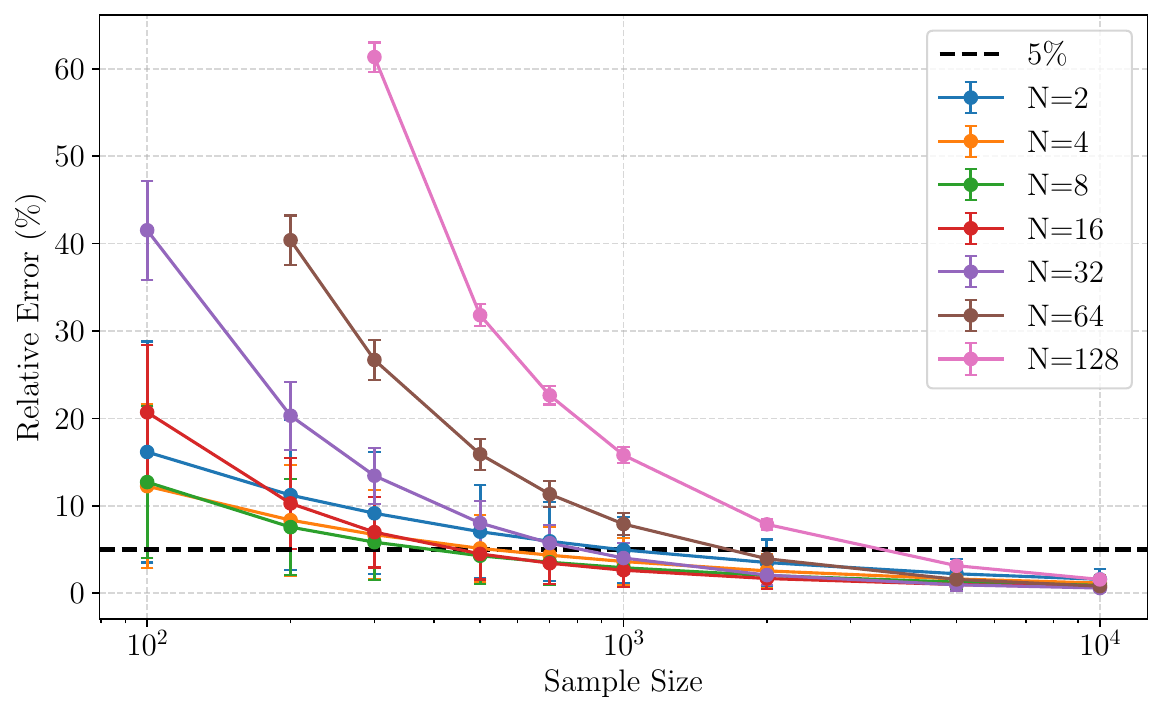}
    \caption{}
    \label{}
    \end{subfigure}
    \hspace*{1mm}
    \begin{subfigure}{0.49\textwidth}
        \centering
        \hspace*{5mm}\figuretitle{Bias correction vs No correction}
        \includegraphics[width=1\linewidth]{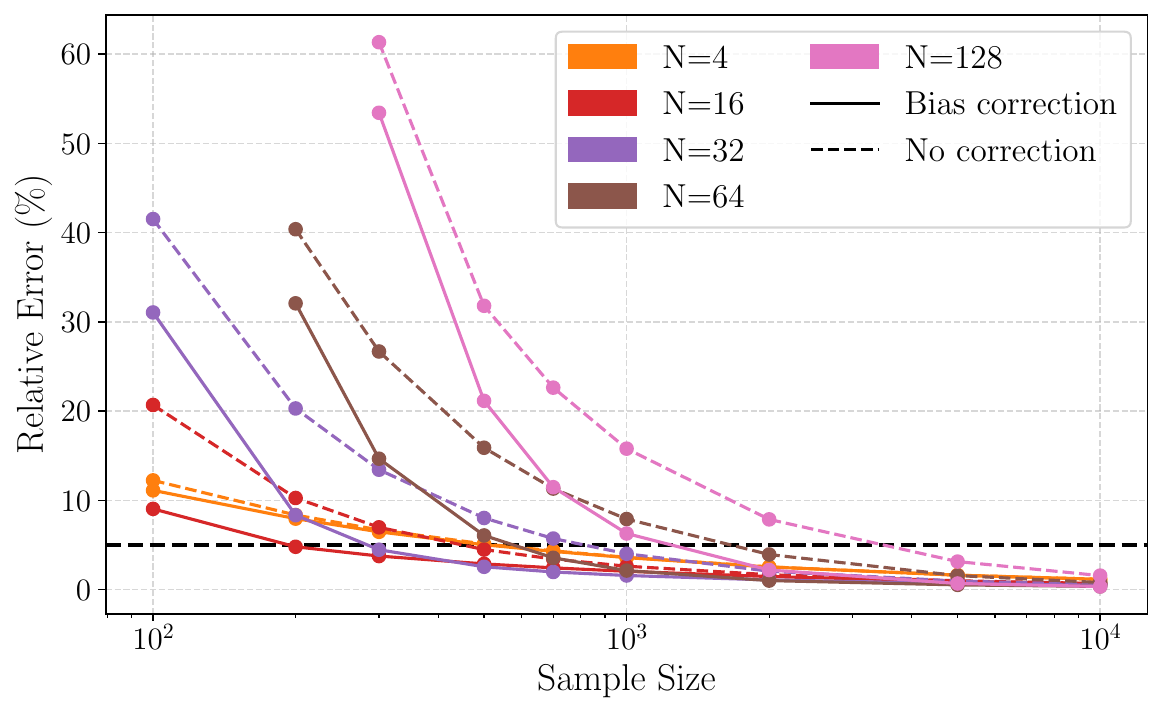}
    \caption{}
    \label{}
    \end{subfigure}
    \caption{(a) Bias estimation of \Winfo against ground truth without bias correction for different system dimensions and sample sizes. Error bars represent the SEM. (b) Comparison of biases with and without bias correction. Error bars were removed for ease of visualisation.}
    \label{fig:no_bias_corr}
\end{figure}

\subsection{High-dimensional autoregressive systems} \label{app:additional_VAR}
Here we extend the analysis presented in Sec.~\ref{sec:results_synthetic} to multivariate time series, studying the behaviour of \Minfo on a larger modular autoregressive model. 
Focusing on $n=10$ time series, we reproduce Experiment 3 by setting the evolution coefficients as 
\begin{equation} \label{eq:VAR_eq3}
    A = 
    \begin{pmatrix}
        \mathcal{A} & \mathcal{B} \\
        \mathcal{B} & \mathcal{A}
    \end{pmatrix} \,,
\end{equation}
where $\mathcal{A}$ is a matrix of size $5 \times 5$ with all entries equal to $a$, and $\mathcal{B}$ is a $5 \times 5$ matrix with all entries equal to $b$, with $a, b \in (0, 1)$. 
The residual covariance matrix $V$ is equal to $c$ on the off-diagonal, and ones on the diagonal, i.e.\ $V_{ij} = c$ for $i \neq j$ and $V_{ii} = 1$. 
As in Sec.~\ref{sec:results_synthetic}, we simulate the dynamics of Eq.~\eqref{eq:var_eq_gen}, rescaling the matrix $A$ so that the system provides the same joint mutual information for each choice of $(a,b)$, and then calculate \W- and \Minfo.

Similarly to the $n=2$ case, we notice that for low values of noise \Minfo is larger when the self and cross-couplings $a$ and $b$ are similar in magnitude, whereas for higher noise correlations \Minfo is maximised for asymmetric interactions (Fig.~\ref{fig:VAR_N10_modular}). 

\begin{figure}[ht]
    \centering
    \hspace*{1mm}\figuretitle{Modular autoregressive model for $\boldsymbol{n=10}$}
    \includegraphics[width=\linewidth]{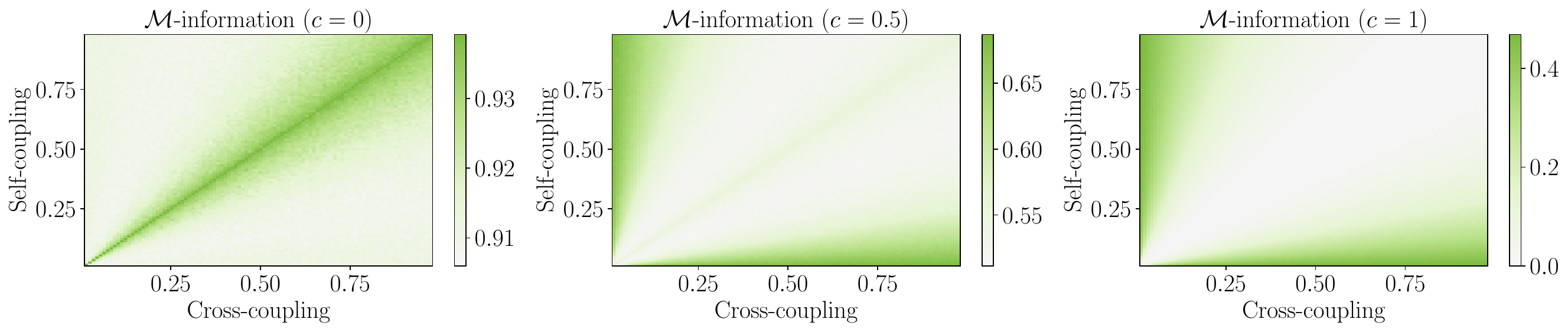}
    \caption{Modular VAR model for $n=10$ time series, with parameters given by Eq.~\eqref{eq:VAR_eq3}. For low values of noise, \Minfo is larger for comparable self and cross-coupling. For higher noise correlations, \Minfo is maximised when either the self-coupling is high and the cross-coupling is low, or vice versa.}
    \label{fig:VAR_N10_modular}
\end{figure}

\begin{figure}
    \centering
    \includegraphics[width=0.9\linewidth]{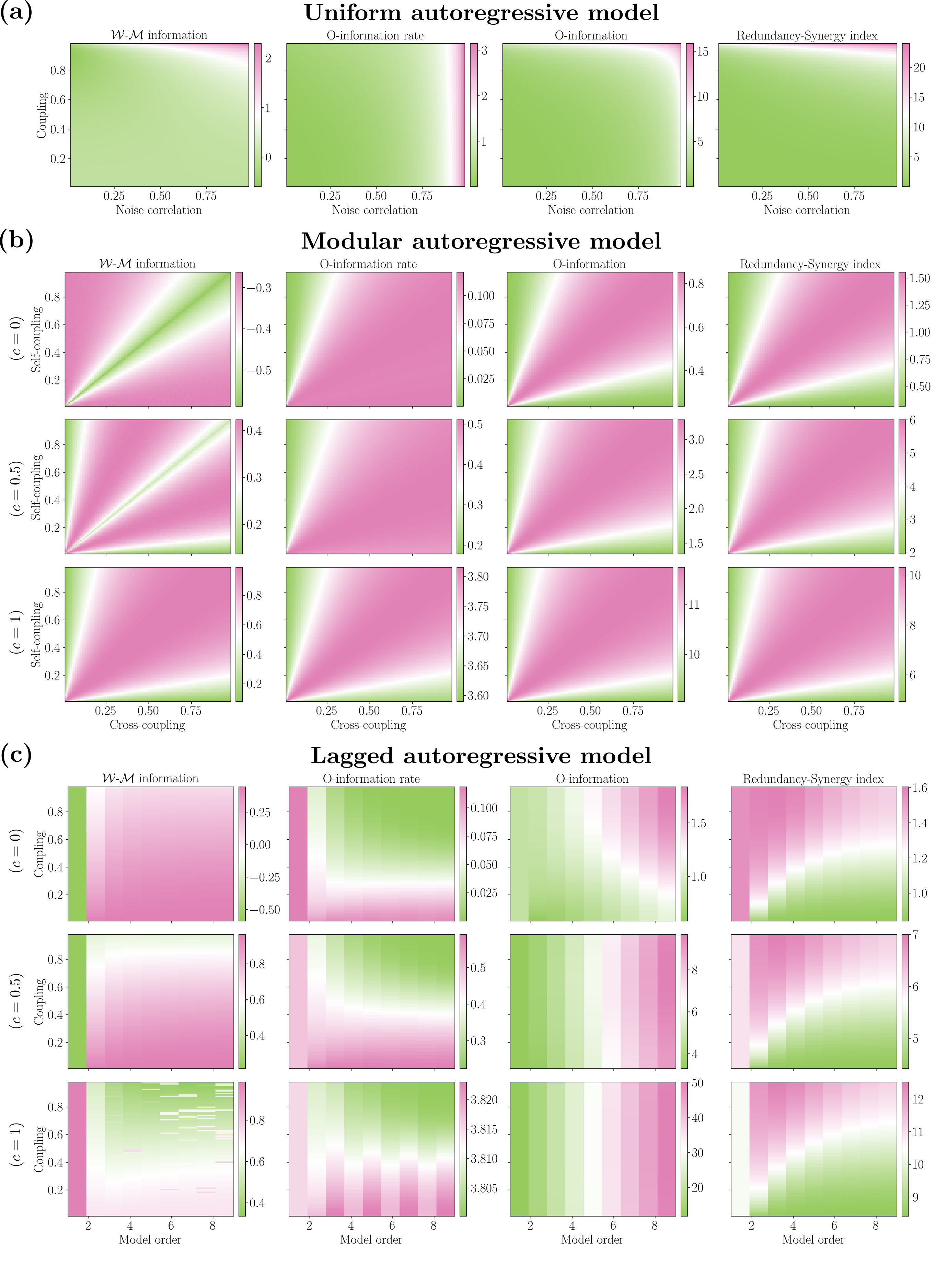}
    \caption{Comparison between \W and \Minfo, O-information rate, O-information, and redundancy-Synergy index on the uniform, modular, and autoregressive models.}
    \label{fig:comparison}
\end{figure}

\subsection{Comparison with existing measures}
In this section, we provide the complete analysis underlying the comparison between \W and \Minfo and existing information-theoretic measures of higher- and lower-order statistical dependencies. In particular, we consider the Organisational information (O-information) \cite{rosas2019quantifying}, its time-resolved extension, the O-information rate (OIR) \cite{faes2022new}, and the Redundancy-Synergy index (RSI) \cite{chechik2001group}. 
Given a multivariate system $(X_1,\ldots,X_n)$, the O-information is defined as the difference between the Total Correlation (TC) and the Dual Total Correlation:
\begin{equation} \label{eq:O-info_def}
\begin{split}
    \Omega(X_1,\ldots,X_n) & := \text{TC}(X_1,\ldots,X_n) - \text{DTC}(X_1,\ldots,X_n) \\
    & = \sum_i^n \left( H(X_i) - H\left( X_i \mid X_{-i}\right) \right) \\
    & = (n-2)H(X_1, \ldots, X_n) - \sum_{i=1}^{n} (H(X_i) - H(X_{-i})) \,,
\end{split}
\end{equation}
where we used the shorthand notation $X_{-i}= (X_1, \ldots, X_{i-1}, X_{i+1}, \ldots, X_n)$.

The O-information rate (OIR) extends this formulation to time-evolving processes, and can be obtained by replacing the entropies in Eq.~\eqref{eq:O-info_def} with entropy rates:
\begin{equation}
     \Omega_r(X_1,\ldots,X_n) := (n-2)H_r(X_1, \ldots, X_n) - \sum_{i=1}^{n} (H_r(X_i) - H_r(X_{-i})) \,,
\end{equation}
where the entropy rate of a process $X_i$ is defined as 
\begin{equation}
    H_r(X_i) = \lim_{t \to \infty} \frac{1}{t} H(X_1^i, X_2^i, \ldots, X_t^i) \,.
\end{equation}

Finally, the Redundancy–Synergy index is a directed measure that quantifies the net balance between redundant and synergistic contributions of a set of sources $(X_1,\ldots,X_n)$ and targets $(Y_1,\ldots,Y_n)$. RSI is defined as
\begin{equation}
\text{RSI}(X_1,\ldots,X_n;Y_1,\ldots,Y_n)
:= \sum_{i,j} I(X_i; Y_j) - I(X_1,\ldots,X_n; Y_1,\ldots,Y_n) \,.
\end{equation}
Negative values of RSI indicate dominance of synergy, while positive contributions suggest a dominant presence of redundancy.

We carry out the analysis on the uniform and modular autoregressive models introduced in Sec.~\ref{sec:results_synthetic}, as well as on an additional dynamical system designed to emphasise the role of temporal structure. We refer to this latter system as the \textit{lagged} autoregressive model, defining it as a VAR($p$) process with variable model order $p$. More specifically, the first lag encodes instantaneous self-coupling only, with $A_{1_{ij}}=0$ for $i \neq j$ and $A_{1_{ii}}=a$, where $a \in (0,1)$. Cross-coupling is instead introduced exclusively at lag $p$, such that $A_{p_{ij}}=a$ for $i \neq j$ and $A_{p_{ii}}=0$, while all remaining coefficient matrices vanish, i.e.\ $A_n=0$ for $n \neq 1,p$. The noise covariance matrix has unit variance on the diagonal ($V_{ii}=1$) and homogeneous correlation $V_{ij}=c$ for $i \neq j$.

Since O-information-based measures are identically zero for bivariate systems, all analyses are conducted on systems of dimension $n=4$. Moreover, as O-information, OIR, and RSI only assess the balance between lower- and higher-order dependencies, we compare them with the difference \W-\M, which,  within our approach, captures the analogous balance. For the modular and lagged scenarios, we consider three representative values of the noise correlation $c$, namely $c\in\{0,0.5,1\}$. The complete findings are presented in Fig.~\ref{fig:comparison}.

Overall, we note that O-information and RSI tend to show similar patterns, which is expected from the close relationship between the two measures \cite{rosas2024characterising}. 
In the uniform model, \W-\M, O-information, and RSI exhibit qualitatively similar dependence on the parameters $a$ and $c$. OIR, however, presents reduced sensitivity in discriminating different coupling strengths $a$, with an average relative change smaller than $2\%$ along this dimension.
In the modular case, \W-\M, O-information, RSI, and OIR all behave similarly for large noise correlation $c$. However, as $c$ decreases, OIR and \W-\M significantly change trend, whereas O-information and RSI present similar patterns $\forall c$, with an average relative variation an order of magnitude smaller than that of OIR and \W-\M. We argue that the behaviour of OIR and \W-\M is expected, as it suggests different information organisations when the noise correlation in the system is absent versus when it dominates, as discussed in Sec~\ref{sec:results_synthetic}. Finally, for small values of $c$, OIR struggles to differentiate configurations with $b>a$. 
In the lagged system, \W-\M and OIR present broadly consistent qualitative trends, with larger $a$ enhancing higher-order structures, though with different susceptibility to the noise correlation parameter.

In sum, we found that \W and \Minfo broadly agree with existing measures, while also providing additional insights in regimes where standard metrics show limited sensitivity. In particular, they can discriminate coupling configurations in the uniform case, distinguish various noise correlations in the modular case, and provide separate estimates of lower- and higher-order information, rather than only their net balance.

\subsection{Additional results on macaques} \label{app:additional_monkey}
In this section, we briefly complement the results reported in the main article with \Minfo computed on the macaque brain across different states of consciousness by examining \Winfo. As in the previous analysis, we compare awake states with sleep, recovery, light and deep anaesthesia.  
Since we consider normalised quantities, it is no surprise that \Winfo shows opposite behaviours to \M. Interestingly, however, while this inverse relationship is indeed observed, differences in \Winfo are not statistically significant. This indicates that, within this framework, variations in levels of consciousness are primarily driven by changes in higher-order interdependencies, rather than by alterations in lower-order organisation.

\begin{figure}[ht]
    \centering
    \includegraphics[width=1\linewidth]{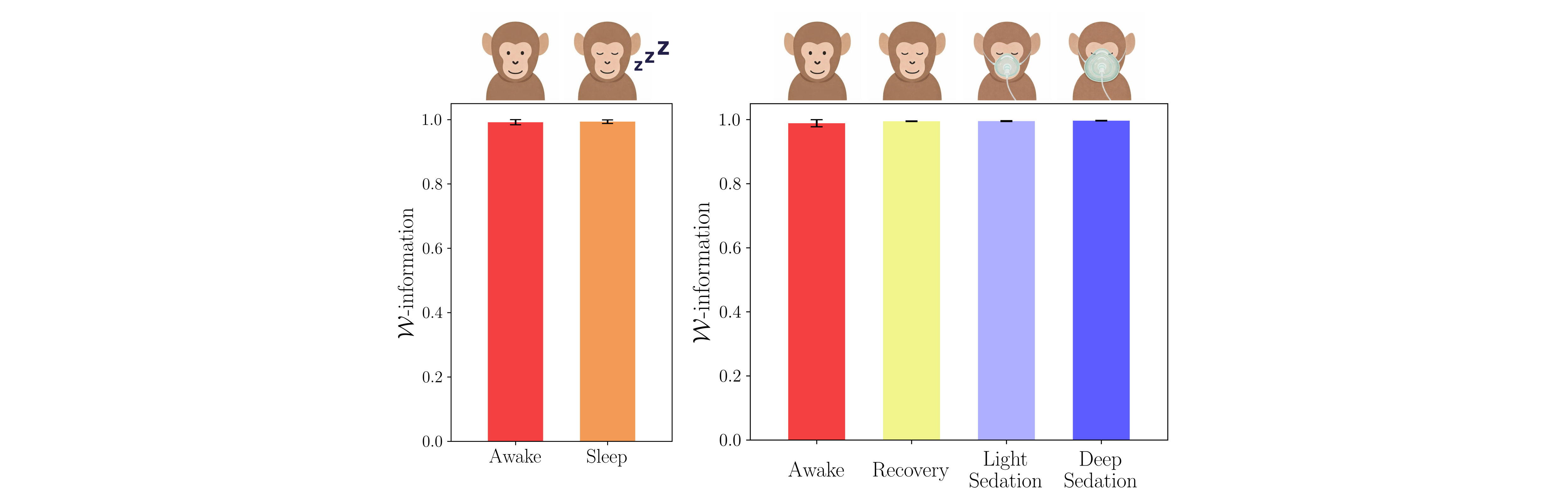}
    \caption{Normalised \Winfo for different states of consciousness in the macaque brain. \W is calculated on all possible pairwise channel combinations and then averaged across pairs. No significant differences are observed between conditions.}
    \label{fig:monkey_app}
\end{figure}

Hence, as mentioned in the main article, these results align with previous findings in the literature \cite{faes2022new} and further clarify the distinct contributions of lower- and higher-order interdependencies.

\subsection{Further analyses on mouse brain} \label{app:additional_mice}
Here, we complement the Neuropixel results presented in the main article with additional insights.
First, we look at the behaviour of Mutual information, \Winfo, and (non-normalised) \Minfo (Fig.~\ref{fig:MI_W_M_info_app}).

\begin{figure}[ht]
    \centering
    \includegraphics[width=1\linewidth]{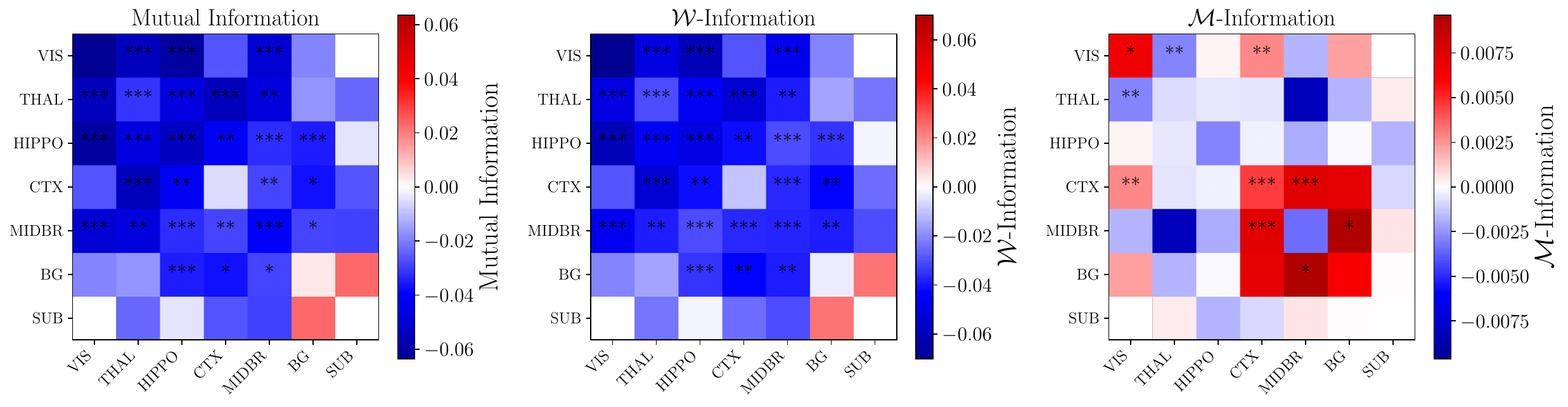}
    \caption{Mutual information, \Winfo, and \Minfo per brain region for correct-incorrect trials in the visual discrimination mice task. (P-values calculated with a one-sample t-test against the zero-mean null hypothesis. *: $p<0.05$; **: $p<0.01$; ***: $p<0.001$).}
    \label{fig:MI_W_M_info_app}
\end{figure}

We notice that the mutual information is significantly higher in the incorrect condition, which can be attributed to a decrease in entropy rate of the neural signal\footnote{The sum of entropy rate and mutual information should yield the total entropy of the time series, which we normalised to 1. Hence, an increase in entropy rate is equivalent to a decrease in mutual information, and vice versa. }. 
We observe consistently lower \Winfo for correct trials, and topographically-specific increases and decreases in \Minfo. 
However, to make more robust comparisons across different systems, a normalisation procedure is needed to factor out the contribution of the mutual information magnitude~\cite{liardi2024null}.  
Hence, in the main body of the work, we showed the values of \Minfo normalised by the joint mutual information. 
To complement those results, here we also report the behaviour of normalised \Winfo, which unexpectedly shows the opposite trend than normalised \Minfo (Fig.~\ref{fig:norm_W_info_app}). 

\begin{figure}[ht]
    \centering
    \includegraphics[width=0.5\linewidth]{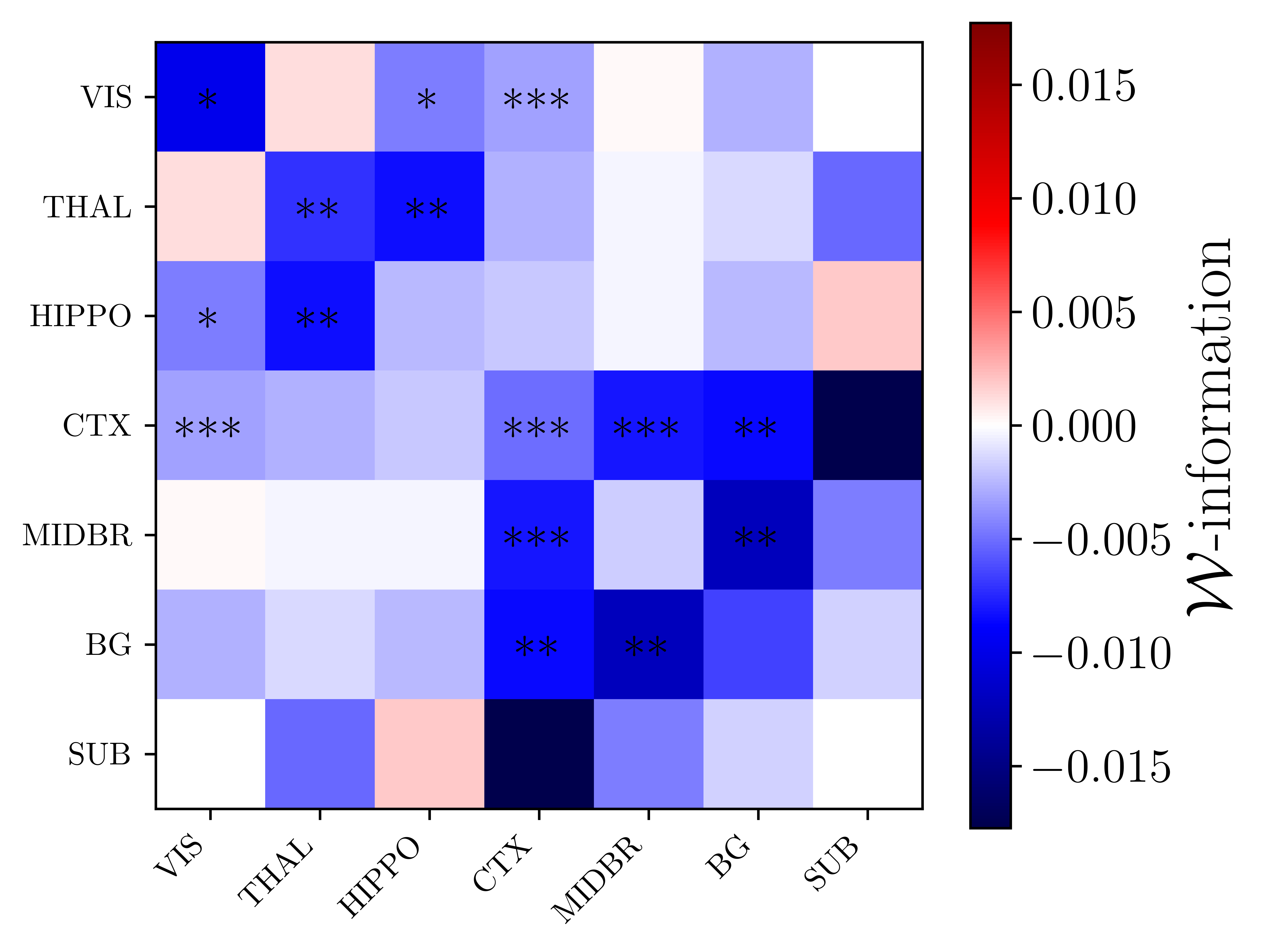}
    \caption{Normalised \Winfo per brain region for correct-incorrect trials. (P-values calculated with a one-sample t-test against the zero-mean null hypothesis. *: $p<0.05$; **: $p<0.01$; ***: $p<0.001$).}
    \label{fig:norm_W_info_app}
\end{figure}

Finally, we focus on the \Minfo calculated between specific theory-driven regions of interest in the mouse brain (Fig.~\ref{fig:specific_ROI}). Namely, we consider:
\begin{itemize}
  \item \textbf{VISp} – Primary Visual Cortex
  \item \textbf{CL} – Central Lateral Nucleus %
  \item \textbf{LGd} – dorsal Lateral Geniculate Nucleus
  \item \textbf{LP} – Lateral Posterior Nucleus %
  \item \textbf{MD} – Mediodorsal Nucleus %
  \item \textbf{POL} – Posterior Lateral Visual Area
  \item \textbf{CA1} – Cornu Ammonis Area 1 %
  \item \textbf{ACA} – Anterior Cingulate Area
  \item \textbf{ILA} – Infralimbic Area
  \item \textbf{RSP} – Retrosplenial Area
\end{itemize}

As expected, we observe an increase in \Minfo between early visual cortex and every other key cortical and subcortical region. Thus, despite our previous observations showing no increase in higher-order thalamo-visual interactions during correct perceptual performance, we find an increase in specific instrumental thalamic nuclei such as the LGd, LP and MD. 
This provides preliminary evidence that low-level sensory subsystems may support perceptually relevant higher-order information structures. Hence, the thalamus, far from being a simple relay station, might possess the capacity for complex macroscopic computation.

We emphasise that these findings are only suggestive of the behaviours discussed, as the small sample size prevented statistical validation. A promising avenue for future work is to rigorously test the hypotheses presented above.  

\begin{figure}[ht]
    \centering
    \includegraphics[width=\linewidth]{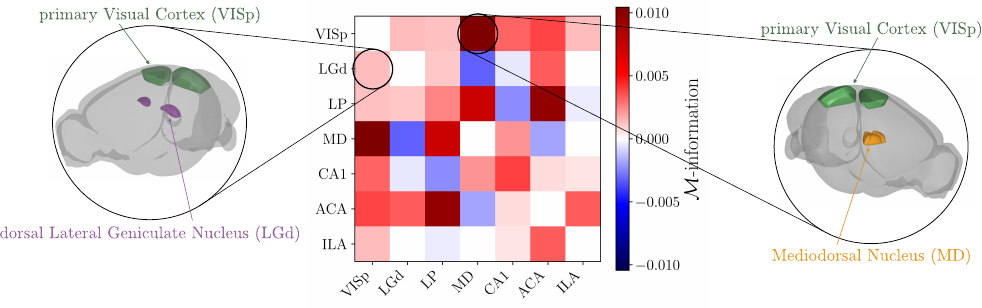}
    \caption{Normalised \Minfo on correct-incorrect trials for specific regions of interest in the mouse brain.%
    }
    \label{fig:specific_ROI}
\end{figure}

\end{document}